\documentclass
[
    a4paper,
    DIV=14,
    abstract=true,
]
{scrartcl}

\usepackage
{
    amsmath,
    amssymb,
    amsthm,
    array,
    authblk,
    dsfont, 
    enumitem,
    graphicx,
    mathtools,
    nicefrac,
    tabularx,
    tcolorbox,
    tikz,
    xcolor,
}

\usepackage[scr=boondoxupr]{mathalfa}

\usepackage[utf8]{inputenc}

\usepackage[pdffitwindow=false,
            plainpages=false,
            pdfpagelabels=true,
            pdfpagemode=UseOutlines,
            pdfpagelayout=SinglePage,
            bookmarks=false,
            colorlinks=true,
            hyperfootnotes=false,
            linkcolor=blue,
            urlcolor=blue!30!black,
            citecolor=green!50!black]{hyperref}

\usepackage[bf,normal]{caption}
\usepackage{color}

\usetikzlibrary{arrows, shapes, fit, automata, positioning}

\graphicspath{{pics/}}

\DeclareMathAlphabet{\mathpzc}{OT1}{pzc}{m}{it}

\newcommand{\subfiguretitle}[1]{{\scriptsize{#1}} \\}
\newcommand{\R}{\mathbb{R}}                                      
\newcommand{\innerprod}[2]{\left\langle #1,\, #2 \right\rangle}  
\newcommand{\ts}{\hspace*{0.1em}}                                
\newcommand{\mc}[2][]{\mathpzc{#2}{\smash[t]{\mathstrut}}_{#1}}  
\newcommand{\relmiddle}[1]{\mathrel{}\middle#1\mathrel{}}        
\providecommand{\abs}[1]{\left\lvert #1 \right\rvert}            
\providecommand{\norm}[1]{\left\lVert #1 \right\rVert}           

\newcommand\xqed[1]{\leavevmode\unskip\penalty9999 \hbox{}\nobreak\hfill \quad\hbox{#1}}
\newcommand{\exampleSymbol}{\xqed{$\triangle$}}

\DeclareMathOperator{\diag}{diag}
\DeclareMathOperator{\tr}{tr}
\DeclareMathOperator{\mspan}{span}
\DeclareMathOperator{\rank}{rank}
\DeclareMathOperator{\Iso}{Iso}
\DeclareMathOperator{\Aut}{Aut}
\let\vec\relax
\DeclareMathOperator{\vec}{vec}

\newtheorem{theorem}{Theorem}[section]
\newtheorem{corollary}[theorem]{Corollary}
\newtheorem{lemma}[theorem]{Lemma}

\newtheorem{definition}[theorem]{Definition}
\theoremstyle{definition}
\newtheorem{example}[theorem]{Example}
\newtheorem{remark}[theorem]{Remark}

\makeatletter
\renewcommand*\env@matrix[1][*\c@MaxMatrixCols c]{%
  \hskip -\arraycolsep
  \let\@ifnextchar\new@ifnextchar
  \array{#1}}
\makeatother

\mathtoolsset{centercolon} 

\allowdisplaybreaks

\DeclareCaptionLabelFormat{period}{#1~#2}
\begin{document}

\title{Continuous optimization methods for the \\ graph isomorphism problem}
\author[1]{Stefan Klus}
\author[2]{Patrick Gel\ss}
\affil[1]{School of Mathematical \& Computer Sciences, Heriot--Watt University, Edinburgh, UK}
\affil[2]{AI in Society, Science, and Technology, Zuse Institute Berlin, Berlin, Germany}

\date{}

\maketitle

\begin{abstract}
The graph isomorphism problem looks deceptively simple, but although polynomial-time algorithms exist for certain types of graphs such as planar graphs and graphs with bounded degree or eigenvalue multiplicity, its complexity class is still unknown. Information about potential isomorphisms between two graphs is contained in the eigenvalues and eigenvectors of their adjacency matrices. However, symmetries of graphs often lead to repeated eigenvalues so that associated eigenvectors are determined only up to basis rotations, which complicates graph isomorphism testing. We consider orthogonal and doubly stochastic relaxations of the graph isomorphism problem, analyze the geometric properties of the resulting solution spaces, and show that their complexity increases significantly if repeated eigenvalues exist. By restricting the search space to suitable subspaces, we derive an efficient Frank--Wolfe based continuous optimization approach for detecting isomorphisms. We illustrate the efficacy of the algorithm with the aid of various highly symmetric graphs.
\end{abstract}

\section{Introduction}

Graphs and networks play a central role in many applications ranging from computer science, engineering, chemistry, and biology to the social sciences. Determining whether or not two graphs are isomorphic, i.e., have exactly the same structure, is called the \emph{graph isomorphism problem} \cite{RC77, ZKT85, Babai19}. It has been shown that the problem can be solved in polynomial time if the graphs are planar \cite{HT72} or have bounded degree~\cite{Luk82} or eigenvalue multiplicity~\cite{BGM82}. Furthermore, Babai recently proved that the graph isomorphism problem is solvable in quasi-polynomial time \cite{Babai16} and summarizes its current status in \cite{Babai19} as ``not expected to be NP-complete, yet not known to be solvable in polynomial time.'' A historical perspective and a concise overview of different techniques for solving the graph isomorphism problem can be found in \cite{GS20}.

It is well known that spectral properties of graphs contain crucial information about potential isomorphisms. If the graphs have distinct eigenvalues, then checking whether a permutation exists that transforms the eigenvectors of one graph into the eigenvectors of the other graph is already sufficient. Repeated eigenvalues, however, complicate graph isomorphism testing since the eigenvectors are then only determined up to basis rotations and it is not possible anymore to directly compare their entries. The graph isomorphism problem can be regarded as a combinatorial optimization problem, which requires minimizing a cost function over the set of all permutation matrices. Different relaxations of the graph isomorphism problem have been proposed in the past, where the set of permutation matrices is either replaced by the set of orthogonal \cite{Umeyama88, ZP08, KS18} or doubly stochastic \cite{ZBV09, ABK15, FS15, Dym18} matrices. The applicability of such relaxations, however, is not well understood \cite{ABK15}. It has been shown in~\cite{ABK15} that the convex relaxation is guaranteed to find the exact isomorphism, provided that the graphs are asymmetric and friendly. This has been generalized to asymmetric graphs with unfriendly eigenvectors that have certain sparsity patterns in~\cite{FS15}. The convex relaxation is called \emph{convex exact} if it does not admit any additional solutions. For which types of symmetries the relaxation is (and in general is not) exact was analyzed in~\cite{Dym18}. We extend these results and consider in particular potentially highly symmetric graphs. We first show how repeated eigenvalues increase the set of feasible solutions of the orthogonal and convex relaxations of the graph isomorphism problem, analyze the geometric properties of these spaces, and then propose an efficient algorithm for detecting graph isomorphisms using a concave reformulation of the relaxed optimization problem that penalizes non-binary matrices. Numerical results for a set of guiding examples and benchmark problems illustrate the efficacy of the proposed algorithm. A method that is similar in spirit is described in \cite{ZBV09}, where a convex--concave relaxation along with a path following approach is used to solve graph matching and quadratic assignment problems. It was recently shown in \cite{Ling24} that if the signal-to-noise ratio is sufficiently large, then the semidefinite relaxations of these problems are exact and recover the true solutions. We focus on the graph isomorphism problem, which allows us to exploit additional information encoded in the null spaces of associated matrices, and construct a direct approach to solve the resulting relaxed optimization problems.

\section{The graph isomorphism problem}

Let $ A $ and $ B $ be the adjacency matrices of two weighted undirected graphs $ \mc[A]{G} = (\mc[A]{V}, \mc[A]{E}) $ and $ \mc[B]{G} = (\mc[B]{V}, \mc[B]{E}) $, respectively, where $ \mc[A]{V} $ and $ \mc[B]{V} $ are the sets of vertices and $ \mc[A]{E} $ and $ \mc[B]{E} $ the sets of edges. We assume that $ \mc[A]{V} = \{ \mc[1]{v}, \dots, \mc[n]{v} \} = \mc[B]{V} $, i.e., the graphs share the same vertices. Our goal is to check if the two graphs are isomorphic.

\begin{definition}[Isomorphic graphs]
Let $ \mathcal{S}(n) $ denote the symmetric group of degree $ n $ and $ \mathcal{P}(n) $ the set of all $ n \times n $ permutation matrices. The graphs $ \mc[A]{G} $ and $ \mc[B]{G} $ are called \emph{isomorphic} if one of the two equivalent conditions is satisfied:
\begin{enumerate}[leftmargin=3.5ex, itemsep=0ex, topsep=0.5ex, label=(\roman*)]
\item There exists a permutation $ \pi \in \mathcal{S}(n) $ such that
\begin{equation*}
    \big(\mc[i]{v}, \mc[j]{v}\big) \in \mc[A]{E} \iff \big(\mc[\pi(i)]{v}, \mc[\pi(j)]{v}\big) \in \mc[B]{E},
\end{equation*}
which preserves the edge weights.
\item There exists a permutation matrix $ P \in \mathcal{P}(n) $ such that
\begin{equation*}
    A = P^\top B \ts P
\end{equation*}
or, equivalently, $ P \ts A = B \ts P $.
\end{enumerate}
\end{definition}

Given the permutation $ \pi $ we can construct the permutation matrix $ P $ and vice versa using
\begin{equation*}
    p_{ij} =
    \begin{cases}
        1, & \text{if } \pi(j) = i, \\
        0, & \text{otherwise}.
    \end{cases}
\end{equation*}
If $ \mc[A]{G} $ and $ \mc[B]{G} $ are isomorphic, we write $ \mc[A]{G} \cong \mc[B]{G} $ and define
\begin{equation*}
    \Iso(\mc[A]{G}, \mc[B]{G}) = \left\{ P \in \mathcal{P}(n) \relmiddle{|} P \ts A = B \ts P \right\}
\end{equation*}
to be the set of all isomorphisms. An isomorphism from a graph $ \mc[A]{G} $ to itself is called an \emph{automorphism} and the set of automorphisms forms a group under composition (or matrix multiplication), which we will denote by $ \Aut(\mc[A]{G}) $. Graphs with nontrivial automorphism groups are called \emph{symmetric}. Examples of highly symmetric graphs are shown in Figure~\ref{fig:Petersen graph}, see also Example~\ref{ex:Petersen graph}.

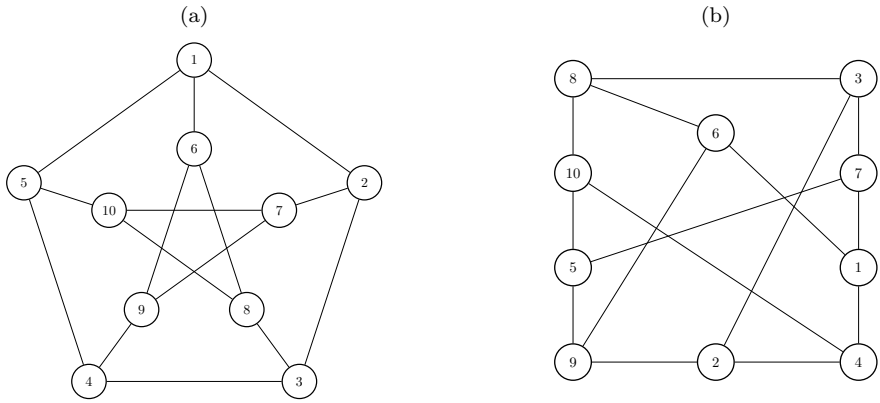
\begin{figure}
    \centering
    \begin{minipage}[t]{0.3\linewidth}
        \centering
        \subfiguretitle{(a)}
        \vspace*{1ex}
        \resizebox{1\textwidth}{!}{%
        \begin{tikzpicture}[
                >= stealth, 
                semithick 
            ]
            \tikzstyle{every state}=[
                draw=black,
                thick,
                fill=white,
                inner sep=0pt,
                text width=6mm,
                align=center,
                scale=0.8
            ]

            \foreach \i in {1,...,5}
                \node[state] at ({4*sin(360*\i/5-72)}, {4*cos(360*\i/5-72)}) (\i) {\i};
            \foreach \i in {6,...,10}
                \node[state] at ({2*sin(360*\i/5-72)}, {2*cos(360*\i/5-72)}) (\i) {\i};
            \path[-] (1) edge (2)
                     (2) edge (3)
                     (3) edge (4)
                     (4) edge (5)
                     (5) edge (1)
                     (1) edge (6)
                     (2) edge (7)
                     (3) edge (8)
                     (4) edge (9)
                     (5) edge (10)
                     (6) edge (8)
                     (6) edge (9)
                     (7) edge (9)
                     (7) edge (10)
                     (8) edge (10);
        \end{tikzpicture}}
    \end{minipage}
    \hspace{10ex}
    \begin{minipage}[t]{0.3\linewidth}
        \centering
        \subfiguretitle{(b)}
        \vspace*{2.4ex}
        \resizebox{0.86\textwidth}{!}{%
        \begin{tikzpicture}[
                >= stealth, 
                semithick 
            ]
            \tikzstyle{every state}=[
                draw=black,
                thick,
                fill=white,
                inner sep=0pt,
                text width=6mm,
                align=center,
                scale=0.8
            ]

            \node[state] at (4, 4) (1) {3};
            \node[state] at (4, 2) (2) {7};
            \node[state] at (4, 0) (3) {1};
            \node[state] at (4, -2) (4) {4};
            \node[state] at (1, -2) (5) {2};
            \node[state] at (-2, 4) (6) {8};
            \node[state] at (-2, 0) (7) {5};
            \node[state] at (1, 2.85) (8) {6};
            \node[state] at (-2, 2) (9) {10};
            \node[state] at (-2, -2) (10) {9};
            \path[-] (1) edge (2)
                     (2) edge (3)
                     (3) edge (4)
                     (4) edge (5)
                     (5) edge (1)
                     (1) edge (6)
                     (2) edge (7)
                     (3) edge (8)
                     (4) edge (9)
                     (5) edge (10)
                     (6) edge (8)
                     (6) edge (9)
                     (7) edge (9)
                     (7) edge (10)
                     (8) edge (10);
        \end{tikzpicture}}
    \end{minipage}
    \caption{(a) Highly symmetric Petersen graph $ \mc[A]{G} $ with $ \abs{\Aut(\mc[A]{G})} = 120 $. (b) A strongly regular graph $ \mc[B]{G} $ with $ 10 $ vertices. Are the two graphs isomorphic?}
    \label{fig:Petersen graph}
\end{figure}

We can formulate the graph isomorphism problem as an optimization problem of the form
\begin{equation} \label{eq:GI}
    c_\mathcal{P}^{} := \min_{P \in \mathcal{P}(n)} \norm{ P \ts A - B \ts P }_F^2,
\end{equation}
where $ \norm{\cdot}_F $ denotes the Frobenius norm, so that $ \mc[A]{G} \cong \mc[B]{G} $ if and only if $ c_\mathcal{P}^{} = 0 $. The solution of the optimization problem, however, is in general not unique if the graphs have symmetries. Since
\begin{equation*}
	\norm{P \ts A - B \ts P}_F^2
        = \norm{A}_F^2 - 2 \ts \tr \left( A^\top \ts P^\top \ts B \ts P \right) + \norm{B}_F^2,
\end{equation*}
minimizing the Frobenius norm in~\eqref{eq:GI} is equivalent to maximizing the trace of $ A^\top \ts P^\top \ts B \ts P $, which illustrates that the graph isomorphism problem or, more generally, the graph matching problem is closely related to the \emph{quadratic assignment problem} \cite{Lawler63, LDBHQ07}. In what follows, we will analyze different relaxations of the graph isomorphism problem.

\section{Orthogonal relaxation}
\label{sec:orthogonal relaxation}

Let $ I_n $ be the $ n \times n $ identity matrix and let
\begin{equation*}
    \mathcal{O}(n) = \left\{ Q \in \R^{n \times n} \relmiddle{|} Q^\top Q = I_n \right\}
\end{equation*}
denote the set of all orthogonal matrices, which forms a nonlinear disconnected manifold (since $ \det(Q) = \pm 1 $). Instead of solving the combinatorial optimization problem \eqref{eq:GI}, we now consider the relaxed problem
\begin{equation} \label{eq:GIO}
    c_\mathcal{O}^{} := \min_{Q \in \mathcal{O}(n)} \norm{Q \ts A - B \ts Q}_F^2,
\end{equation}
which is also known as the \emph{two-sided orthogonal Procrustes problem}~\cite{Sch68, GD04}.

\begin{definition}[Eigendecomposition] \label{def:eigenvector partitioning}
Given a symmetric matrix $ A \in \R^{n \times n} $, let $ A = U_A \ts \Lambda_A \ts U_A^\top $ be its eigendecomposition, where the columns of $ U_A $ are the eigenvectors and $ \Lambda_A $ is a diagonal matrix containing the eigenvalues sorted in non-increasing order. We define $ \lambda_A = \big[ \lambda_A^{(1)}, \dots, \lambda_A^{(m)} \big] $ to be the row vector containing the $ m \le n $ unique eigenvalues sorted in decreasing order and $ \mu_A = \big[ \mu_A^{(1)}, \dots, \mu_A^{(m)} \big] $ to be the row vector containing the corresponding multiplicities. That is, it holds that
\begin{equation*}
    \sum_{k=1}^m \mu_A^{(k)} = n.
\end{equation*}
We then partition $ U_A $ into
\begin{equation*}
    U_A = \Big[ U_A^{(1)}, \dots, U_A^{(m)} \Big], \quad \text{with} \quad U_A^{(k)} = \Big[u_A^{(k, 1)}, \dots, u_A^{(k, \mu_A^{(k)})}\Big] \in \R^{n \times \mu_A^{(k)}}.
\end{equation*}
\end{definition}

If $ \lambda_A^{(k)} $ is a distinct eigenvalue, then $ U_A^{(k)} $ is the associated eigenvector. If, on the other hand, $ \lambda_A^{(k)} $ is a repeated eigenvalue, then the columns of $ U_A^{(k)} $ form an orthogonal basis of the corresponding eigenspace. Repeated eigenvalues can be caused by symmetries of the graph. Assume that $ \big(\lambda_A^{(k)}, u_A^{(k)}\big) $ is an eigenvalue--eigenvector pair and let $ P \in \Aut(\mc[A]{G}) $, i.e., $ P \ts A = A \ts P $, then
\begin{equation*}
    A \ts u_A^{(k)} = \lambda_A^{(k)} \ts u_A^{(k)} \implies A \ts P \ts u_A^{(k)} = \lambda_A^{(k)} \ts P \ts u_A^{(k)}
\end{equation*}
and $ P \ts u_A^{(k)} $ is also an eigenvector associated with $ \lambda_A^{(k)} $. If these two eigenvectors are linearly independent, this implies that the eigenspace is at least two-dimensional. The larger the automorphism group, the higher the likelihood that the associated permuted eigenvectors are linearly independent. It was shown in \cite{Mowshowitz69, Lov07} that if all eigenvalues of a graph are distinct, then every non-trivial automorphism has order $ 2 $. Nevertheless, repeated eigenvalues are not necessarily related to symmetries. The graph in Figure~\ref{fig:Example graphs}\ts(a), for example, is asymmetric but has repeated eigenvalues. Since it is in practice difficult to determine the size of the automorphism group, a related property called \emph{friendliness}, which can be verified easily, was proposed in \cite{ABK15}.

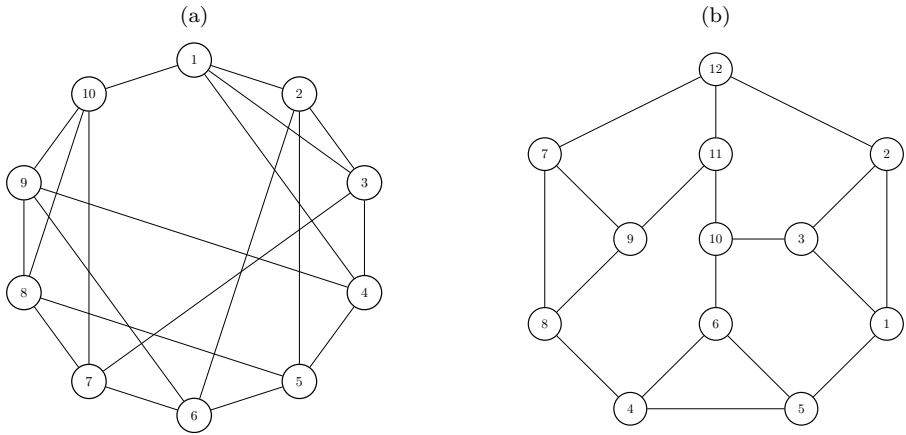
\begin{figure}
    \centering
    \begin{minipage}[t]{0.3\linewidth}
        \centering
        \subfiguretitle{(a)}
        \vspace*{1ex}
        \resizebox{1\textwidth}{!}{%
        \begin{tikzpicture}[
                >= stealth, 
                semithick 
            ]
            \tikzstyle{every state}=[
                draw=black,
                thick,
                fill=white,
                inner sep=0pt,
                text width=6mm,
                align=center,
                scale=0.8
            ]

            \foreach \i in {1,...,10}
                \node[state] at ({4*sin(360*\i/10-36)}, {4*cos(360*\i/10-36)}) (\i) {\i};

            \path[-] (1) edge (2)
                     (1) edge (3)
                     (1) edge (4)
                     (1) edge (10)
                     (2) edge (3)
                     (2) edge (5)
                     (2) edge (6)
                     (3) edge (4)
                     (3) edge (7)
                     (4) edge (5)
                     (4) edge (9)
                     (5) edge (6)
                     (5) edge (8)
                     (6) edge (7)
                     (6) edge (9)
                     (7) edge (8)
                     (7) edge (10)
                     (8) edge (9)
                     (8) edge (10)
                     (9) edge (10);
        \end{tikzpicture}}
    \end{minipage}
    \hspace{10ex}
    \begin{minipage}[t]{0.3\linewidth}
        \centering
        \subfiguretitle{(b)}
        \vspace*{1.8ex}
        \resizebox{1\textwidth}{!}{%
        \begin{tikzpicture}[
                semithick 
            ]
            \tikzstyle{every state}=[
                draw=black,
                thick,
                fill=white,
                inner sep=0pt,
                text width=6mm,
                align=center,
                scale=0.8
            ]

            \node[state] (v1) at (4, -2) {1};
            \node[state] (v2) at (4, 2) {2};
            \node[state] (v3) at (2, 0) {3};
            \node[state] (v4) at (-2, -4) {4};
            \node[state] (v5) at (2, -4) {5};
            \node[state] (v6) at (0, -2) {6};
            \node[state] (v7) at (-4, -2) {8};
            \node[state] (v8) at (-2, 0) {9};
            \node[state] (v9) at (-4, 2) {7};
            \node[state] (v10) at (0, 0) {10};
            \node[state] (v11) at (0, 4) {12};
            \node[state] (v12) at (0, 2) {11};

            \path[-] (v1) edge (v2)
                     (v2) edge (v3)
                     (v3) edge (v1)
                     (v4) edge (v5)
                     (v5) edge (v6)
                     (v6) edge (v4)
                     (v7) edge (v8)
                     (v8) edge (v9)
                     (v9) edge (v7)
                     (v8) edge (v12)
                     (v6) edge (v10)
                     (v3) edge (v10)
                     (v11) edge (v12)
                     (v10) edge (v12)
                     (v4) edge (v7)
                     (v9) edge (v11)
                     (v11) edge (v2)
                     (v1) edge (v5);
        \end{tikzpicture}}
    \end{minipage}
    \caption{(a) Asymmetric unfriendly regular graph with repeated eigenvalues. (b) Asymmetric unfriendly regular Frucht graph with distinct eigenvalues.}
    \label{fig:Example graphs}
\end{figure}

\begin{definition}[Friendly eigenvectors] \label{def:friendly eigenvectors}
An eigenvector $ u_A^{(k)} $ is called \emph{friendly} if $ \big\langle u_A^{(k)}, \mathds{1}_n \big\rangle \ne 0 $, where $ \mathds{1}_n \in \R^n $ denotes the vector of all ones. Correspondingly, a graph $ \mc[A]{G} $ (or its adjacency matrix $ A $) is called \emph{friendly} if all eigenvalues $ \lambda_A^{(k)} $ are distinct and all eigenvectors are friendly.
\end{definition}

Every friendly graph is automatically asymmetric as shown in \cite{ABK15}. The other direction, however, is not true. The Frucht graph \cite{Frucht39} illustrated in Figure~\ref{fig:Example graphs}\ts(b), for instance, is asymmetric but not friendly, see also \cite{FS15}. In fact, any asymmetric regular graph will be highly unfriendly since $ u_A^{(1)} = \mathds{1}_n $ is an eigenvector and the remaining $ n-1 $ eigenvectors are orthogonal to it. In \cite{KS18}, we defined a stronger property, which can sometimes be used to disambiguate the signs of unfriendly eigenvectors.

\begin{definition}[Ambiguous eigenvectors]
We define an eigenvector $ u_A^{(k)} $ to be \emph{ambiguous} if there exists $ P \in \mathcal{P}(n) $ with $ u_A^{(k)} = -P u_A^{(k)} $, i.e., $ u_A^{(k)} $ and $ -u_A^{(k)} $ share the same entries. If all eigenvalues $ \lambda_A^{(k)} $ are distinct and none of the associated eigenvectors are ambiguous, we call the graph $ \mc[A]{G} $ \emph{unambiguous}.
\end{definition}

This definition allows us to assign a canonical sign to unambiguous eigenvectors. An ambiguous eigenvector is automatically unfriendly since $ \big\langle u_A^{(k)}, \mathds{1}_n \big\rangle = -\big\langle P \ts u_A^{(k)}, \mathds{1}_n \big\rangle = -\big\langle u_A^{(k)}, \mathds{1}_n \big\rangle $, which implies $ \big\langle u_A^{(k)}, \mathds{1}_n \big\rangle = 0 $. The vector $ u_A^{(k)} = [1, -2, 0, 1]^\top $, for instance, is unfriendly but not ambiguous.

\begin{definition}[Isospectral graphs]
We call two symmetric matrices $ A $ and $ B $ and hence also the induced graphs $ \mc[A]{G} $ and $ \mc[B]{G} $ \emph{isospectral} if $ \lambda_A = \lambda_B $ and $ \mu_A = \mu_B $.
\end{definition}

That is, isospectral graphs have the same eigenvalues and also their multiplicities are identical. Isomorphic graphs are clearly isospectral.

\begin{theorem} \label{thm:Procrustes}
Assume that $ A $ and $ B $ are isospectral. Let $ R \in \R^{n \times n} $ be a block diagonal matrix of the form
\begin{equation*}
    R =
    \begin{bmatrix}
        R^{(1)} & \\
        & \ddots \\
        & & R^{(m)}
    \end{bmatrix},
\end{equation*}
where $ R^{(k)} \in \mathcal{O}\big(\mu_A^{(k)}\big) $. Then
\begin{equation*}
    Q^*(R) = U_B \ts R \ts U_A^\top
\end{equation*}
minimizes the cost function \eqref{eq:GIO} and all minimizers are of this form.
\end{theorem}

\begin{proof}
First, $ Q^*(R) $ is a product of orthogonal matrices and hence orthogonal. We have
\begin{equation*}
\renewcommand*{\arraystretch}{1.4}
\begin{split}
    R \ts \Lambda_A &=
    \begin{bmatrix}
        R^{(1)} & \\
        & \ddots \\
        & & R^{(m)}
    \end{bmatrix}
    \begin{bmatrix}
        \lambda_A^{(1)} I_{\mu_A^{(1)}} & \\
        & \ddots \\
        & & \lambda_A^{(m)} I_{\mu_A^{(m)}}
    \end{bmatrix} \\
    &=
    \begin{bmatrix}
        \lambda_A^{(1)} I_{\mu_A^{(1)}} & \\
        & \ddots \\
        & & \lambda_A^{(m)} I_{\mu_A^{(m)}}
    \end{bmatrix}
    \begin{bmatrix}
        R^{(1)} & \\
        & \ddots \\
        & & R^{(m)}
    \end{bmatrix}
    = \Lambda_A R
\end{split}
\end{equation*}
and since $ A $ and $ B $ are by definition isospectral also $ R \ts \Lambda_A = \Lambda_B R $. It follows that
\begin{align*}
    Q^*(R) \ts A - B \ts Q^*(R)
        &= U_B \ts R \ts U_A^\top A - B \ts U_B \ts R \ts U_A^\top \\
        &= U_B \ts R \ts \Lambda_A U_A^\top - U_B \ts \Lambda_B \ts R \ts U_A^\top \\
        &= U_B \big(R \ts \Lambda_A - \Lambda_B \ts R\big) U_A^\top \\
        &= 0.
\end{align*}
That is, $ Q^*(R) $ minimizes the cost function and $ c_\mathcal{O}^{} = 0 $. To show that all minimizers must be of this form, assume now that $ c_\mathcal{O}^{} = 0 $, then
\begin{equation*}
    B = Q \ts A \ts Q^\top = Q \ts U_A \ts \Lambda_A \ts U_A^\top \ts Q^\top = (Q \ts U_A) \ts \Lambda_B \ts (Q \ts U_A)^\top,
\end{equation*}
i.e., the columns of $ Q \ts U_A $ are eigenvectors of $ B $, which are determined up to basis rotations (in the corresponding eigenspaces). Writing
\begin{equation*}
    (Q \ts U_A) = \Big[ \widetilde{U}_B^{(1)}, \dots, \widetilde{U}_B^{(m)} \Big],
\end{equation*}
it follows that $ \widetilde{U}_B^{(k)} = U_B^{(k)} \ts R^{(k)} $, where $ R^{(k)} \in \mathcal{O}\big(\mu_A^{(k)}\big) $. Overall, this implies that $ Q \ts U_A = U_B \ts R $ and thus $ Q = U_B \ts R \ts U_A^\top $, where $ R $ is an orthogonal block diagonal matrix.
\end{proof}

Note that we did not explicitly use the orthogonality of the matrices $ R^{(k)} $, which is only required so that $ Q^*(R) $ remains orthogonal. A similar proof can be found in \cite{LM79}. For a more detailed discussion of the equivariance properties of self-adjoint matrices, see also \cite{DGGK19}.

\begin{remark} We have to distinguish between different cases:
\begin{enumerate}[leftmargin=3.5ex, itemsep=0ex, topsep=0.5ex, label=(\roman*)]
\item Assume that all eigenvalues of $ A $ and $ B $ are distinct, then $ m = n $ and
\begin{equation*}
    R =
    \begin{bmatrix}
        \pm 1 \\
        & \ddots \\
        & & \pm 1
    \end{bmatrix},
\end{equation*}
i.e., there exist $ 2^n $ orthogonal matrices $ Q^*(R) $ that minimize the relaxed cost function. This is the special case typically considered in the literature, see, e.g., \cite{Sch68, Umeyama88}.
\item If the eigenvalues are not all distinct, then there are
\begin{equation*}
    \sum_{k=1}^m \frac{\mu_A^{(k)} \big(\mu_A^{(k)} - 1\big)}{2}
\end{equation*}
(continuous) degrees of freedom\footnote{The set $ \mathcal{O}\big(\mu_A^{(k)}\big) $ is a compact Lie group of dimension $ \frac{\mu_A^{(k)} \big(\mu_A^{(k)} - 1\big)}{2} $.} and the search space, i.e., the set of all minimizers of~\eqref{eq:GIO}, is significantly more complex.
\end{enumerate}
\end{remark}

A simple way to check whether two graphs are potentially isomorphic is to first solve the orthogonal Procrustes problem \eqref{eq:GIO}. If $ c_\mathcal{O}^{} > 0 $, the graphs cannot be isomorphic. However, if $ c_\mathcal{O}^{} = 0 $, this only implies that the graphs are isospectral (since $ A $ and $ B $ are in this case similar) but not necessarily isomorphic (i.e., $ c_\mathcal{P}^{} $ might still be greater than zero). The graph isomorphism problem can thus be reformulated as: Does the set of feasible solutions $ Q^*(R) $ contain at least one permutation matrix? Finding such a matrix, assuming the graphs are isomorphic, is a difficult nonlinear optimization problem. One potential way to obtain a valid solution of the graph isomorphism problem is to penalize orthogonal matrices that are not permutation matrices. This approach was used in \cite{ZP08} to solve the graph matching problem.

\section{Doubly stochastic relaxation}

We will now consider a different relaxation of the graph isomorphism problem. Let
\begin{equation*}
    \mathcal{D}(n) = \left\{ X \in \R^{n \times n} \relmiddle{|} X \ge 0, X \mathds{1}_n = \mathds{1}_n, \mathds{1}_n^\top X = \mathds{1}_n^\top \right\}
\end{equation*}
be the set of doubly stochastic matrices, which is the convex hull of the set of all permutation matrices~$ \mathcal{P}(n) $. This is also known as \emph{Birkhoff's theorem} \cite{Birkhoff46}. The doubly stochastic relaxation of the graph isomorphism problem is given by
\begin{equation} \label{eq:GIDS}
    c_\mathcal{D}^{} := \min_{X \in \mathcal{D}(n)} \norm{X \ts A - B \ts X}_F^2.
\end{equation}
We again have the property that the graphs cannot be isomorphic if $ c_\mathcal{D}^{} > 0 $.

\subsection{Solutions of the relaxed problem}

The convex relaxation of the problem admits additional solutions as the following lemma shows.

\begin{lemma} \label{lem:convex combinations}
Given two isomorphic graphs $ \mc[A]{G} \cong \mc[B]{G} $. Let $ P_1, \dots, P_\ell \in \Iso(\mc[A]{G}, \mc[B]{G}) $, then any convex combination of the form
\begin{equation*}
    X = \sum_{k=1}^\ell \gamma_k \ts P_k, \quad \text{with }\sum_{k=1}^\ell \gamma_k = 1,
\end{equation*}
minimizes \eqref{eq:GIDS}.
\end{lemma}

\begin{proof}
A convex combination of doubly stochastic matrices is again doubly stochastic so that $ X \in \mathcal{D}(n) $. We have
\begin{equation*}
    X \ts A - B \ts X = \sum_{k=1}^\ell \gamma_k \big(P_k \ts A - B \ts P_k\big) = 0
\end{equation*}
as required.
\end{proof}

It is important to note here that all convex combinations of isomorphisms are solutions of the relaxed problem, but not all doubly stochastic matrices that minimize \eqref{eq:GIDS} are automatically convex combinations of isomorphisms as the following example shows.

\begin{example} \label{ex:Frucht graph feasbible solutions}
As mentioned above, the automorphism group of the Frucht graph contains only the identity matrix, but $ X^* = \frac{1}{n} \mathds{1} \mathds{1}^\top $ is a solution of the relaxed problem. We will see that there is in fact an entire $ 11 $-dimensional affine subspace of feasible doubly stochastic solutions. \exampleSymbol
\end{example}

This is related to the aforementioned friendliness property, but most of the eigenvectors are in this case not ambiguous so that we can eliminate spurious solutions as we will show below.

\subsection{Vectorization of the problem}

We now rewrite the optimization problem \eqref{eq:GIDS} as
\begin{equation} \label{eq:GIDS trace}
\begin{split}
    \norm{X \ts A - B \ts X}_F^2
        &= \tr\left( A^2 \ts X^\top X \right) - 2 \tr\left(A \ts X^\top B \ts X \right) + \tr\left(X^\top B^2 \ts X \right)
\end{split}
\end{equation}
using the cyclic property of the trace and the fact that $ A $ and $ B $ are symmetric. Note that here the first term and last term are not independent of $ X $ as they are for the orthogonal relaxation since in general $ X^\top X \ne I_n \ne X \ts X^\top $.

\begin{definition}[Vectorization]
Given a matrix $ X \in \R^{n \times n} $, the corresponding \emph{vectorized} matrix $ x \in \R^{n^2} $ is defined by
\begin{equation*}
    x = \vec(X) =
    \begin{bmatrix}
        X_1 \\
        \vdots \\
        X_n
    \end{bmatrix},
\end{equation*}
where $ X_1, \dots, X_n $ denote the columns of $ X $.
\end{definition}

\begin{lemma} \label{lem:trace as quadratic form}
Let $ A, B, X \in \R^{n \times n} $, then
\begin{equation*}
    \tr(A^\top \ts X^\top B \ts X) = x^\top (A \otimes B) \ts x,
\end{equation*}
where $ x = \vec(X) $ and $ \otimes $ denotes the Kronecker product.
\end{lemma}

\begin{proof}
We include the proof for the sake of completeness. The entries of $ A \otimes B $ are by definition given by
\begin{equation*}
	(A \otimes B)_{(i_1-1) n + i_2, \ts (j_1-1) n + j_2} = A_{i_1, j_1} \ts B_{i_2,j_2},
\end{equation*}
whereas the entries of the vectorization of $ X $, denoted by $ x $, and $ X $ itself are related by
\begin{equation*}
	x_{(j-1) n + i} = \vec(X)_{(j-1) n + i} = X_{i,j}.
\end{equation*}
These index mappings can be viewed as row- and column-major orderings, respectively. By combining the above relations, we obtain
\begin{equation*}
	\begin{split}
		\tr(A^\top \ts X^\top B \ts X) &= \sum_{i_1=1}^n \sum_{i_2=1}^n \sum_{j_1=1}^n \sum_{j_2=1}^n A^\top_{j_1,i_1} \ts X^\top_{i_1, i_2} \ts B_{i_2,j_2} \ts X_{j_2,j_1} \\
		&= \sum_{i_1=1}^n \sum_{i_2=1}^n \sum_{j_1=1}^n \sum_{j_2=1}^n X_{i_2, i_1} \ts A_{i_1,j_1} \ts B_{i_2,j_2} \ts X_{j_2,j_1} \\
		&= \sum_{i_1=1}^n \sum_{i_2=1}^n \sum_{j_1=1}^n \sum_{j_2=1}^n x_{(i_1-1)n + i_2} \ts (A \otimes B)_{(i_1-1) n + i_2, \ts (j_1-1) n + j_2} \ts x_{(j_1-1)n + j_2} \\ &= \sum_{i=1}^{n^2} \sum_{j=1}^{n^2} x_i \ts (A \otimes B)_{i, j} \ts x_j,
	\end{split}
\end{equation*}
which concludes the proof.
\end{proof}

If $ A $ is symmetric, we have $ \tr(A \ts X^\top B \ts X) = x^\top (A \otimes B) \ts x $. The row and column sum constraints can be reformulated as $ C \ts x = d $, with
\begin{equation*}
    C =
    \begin{bmatrix}
        \mathds{1}_n^\top \\
        & \ddots \\
        & & \mathds{1}_n^\top \\[0.5ex]
        e_1^\top & \dots & e_1^\top \\
        \vdots & \ddots & \vdots \\
        e_n^\top & \dots & e_n^\top \\
    \end{bmatrix} \in \R^{2 \ts n \times n^2}
\end{equation*}
and $ d = \mathds{1}_{2 \ts n} $, where $ e_i \in \R^n $ denotes the $ i $th unit vector. This defines a $ (2 \ts n - 1) $-dimensional affine subspace.

\begin{lemma} \label{lem:quadratic form}
The optimization problem \eqref{eq:GIDS} can be written as
\begin{equation*}
    \min_{\substack{x \ge 0 \\[0.4ex] C x = d}} x^\top H x,
\end{equation*}
with
\begin{equation*}
    H = \big(A^2 \otimes I_n\big) - 2 \ts \big(A \otimes B\big) + \big(I_n \otimes B^2\big) = \left( A \otimes I_n - I_n \otimes B \right)^2,
\end{equation*}
where again $ x = \vec(X) $.
\end{lemma}

\begin{proof}
The result follows immediately from Lemma~\ref{lem:trace as quadratic form} and the trace formulation \eqref{eq:GIDS trace}.
\end{proof}

This is a convex but in general not strictly convex optimization problem. In particular, all matrices $ X $ defined in Lemma~\ref{lem:convex combinations} are optimal solutions.

\begin{theorem} \label{thm:properties of H}
It holds that:
\begin{enumerate}[leftmargin=3.5ex, itemsep=0ex, topsep=0.5ex, label=(\roman*)]
\item The matrix $ H $ is symmetric and positive semi-definite.
\item Let $ \lambda_A^{(i)} $ and $ \lambda_B^{(j)} $ be eigenvalues of $ A $ and $ B $, respectively, and $ u_A^{(i, \alpha)} $ and $ u_B^{(j, \beta)} $ associated eigenvectors, where $ 1 \le \alpha \le \mu_A^{(i)} $ and $ 1 \le \beta \le \mu_B^{(j)} $, then
\begin{equation*}
    \lambda_H^{(i,j)} = \big(\lambda_A^{(i)} - \lambda_B^{(j)} \big)^2
\end{equation*}
is an eigenvalue of $ H $ and the corresponding eigenvector is given by
\begin{equation*}
    u_H^{(i,\alpha,j,\beta)} = u_A^{(i,\alpha)} \otimes u_B^{(j,\beta)}.
\end{equation*}
\end{enumerate}
\end{theorem}

\begin{proof} ~
\begin{enumerate}[leftmargin=3.5ex, itemsep=0ex, topsep=0.5ex, label=(\roman*)]
\item All three terms of $ H $ are Kronecker products of symmetric matrices and thus also symmetric and so is their sum. The matrix $ H $ is positive semi-definite since
\begin{equation*}
    x^\top H x = \norm{X \ts A - B \ts X}_F^2 \ge 0
\end{equation*}
for all $ x $. Alternatively, this also follows from (ii).
\item We have
\begin{align*}
	H \ts u_H^{(i,\alpha,j,\beta)} &= \big( A \otimes I_n - I_n \otimes B\big)^2 \big(u_A^{(i,\alpha)} \otimes u_B^{(j,\beta)}\big) \\
	&= \big( A \otimes I_n - I_n \otimes B \big) \big(\lambda_A^{(i)} \ts u_A^{(i,\alpha)} \otimes u_B^{(j,\beta)} - \lambda_B^{(j)} \ts u_A^{(i,\alpha)} \otimes u_B^{(j,\beta)} \big) \\
	&= \big(\lambda_A^{(i)} - \lambda_B^{(j)}\big) \ts \big( A \otimes I_n - I_n \otimes B \big) \big( u_A^{(i,\alpha)} \otimes u_B^{(j,\beta)}\big) \\
	&= \big(\lambda_A^{(i)} - \lambda_B^{(j)}\big)^2 \ts \big( u_A^{(i,\alpha)} \otimes u_B^{(j,\beta)}\big) \\
	&= \lambda_H^{(i,j)} \ts u_H^{(i,\alpha,j,\beta)}.
\end{align*}
That is, all eigenvalues of $ H $ are non-negative. \qedhere
\end{enumerate}
\end{proof}

\subsection{Isospectral graphs}

If $ A $ and $ B $ are isospectral, Theorem~\ref{thm:properties of H} immediately implies that the rank of $ H $ is at most $ n (n-1) $ since at least $ n $ of the eigenvalues of $ H $ are zero. If the matrices are isospectral and also have repeated eigenvalues, the rank will be accordingly lower.

\begin{corollary} \label{cor:rank of H}
Assume that $ A $ and $ B $ are isospectral. Let $ \lambda_A = \lambda_B $ be the eigenvalues with multiplicities $ \mu_A = \mu_B $ (see Definition~\ref{def:eigenvector partitioning}). Then
\begin{equation*}
     \rank(H) = n^2 - \sum_{k=1}^m \big(\mu_A^{(k)} \big)^2.
\end{equation*}
Furthermore, the null space of the matrix $ H $ is given by
\begin{equation*}
    \operatorname{null}(H) = \mspan\left\{ u_A^{(k,\alpha)} \otimes u_B^{(k,\beta)} \relmiddle{|} 1 \le \alpha, \beta \le \mu_{A}^{(k)}, k = 1, \dots, m \right\}.
\end{equation*}
\end{corollary}

\begin{proof}
For each $ \lambda_A^{(k)} = \lambda_B^{(k)} $ there are $ \mu_A^{(k)} = \mu_B^{(k)} $ linearly independent eigenvectors and hence $ \big(\mu_A^{(k)}\big)^2 $ possible combinations for which the associated eigenvalues of $ H $ are zero. The corresponding eigenvectors span the null space of $ H $. The result then follows from the rank--nullity theorem.
\end{proof}

\begin{example}
Assume that $ \lambda_A = [2, 1, 0] = \lambda_B $ with multiplicities $ \mu_A = [2, 1, 3] = \mu_B $, then the eigenvalues of~$ H $, written in matrix form, are given by
\begin{align*}
    & \hspace*{2.4ex} \scalebox{0.8}{2} \hspace*{2.4ex} \scalebox{0.8}{2} \hspace*{2.4ex} \scalebox{0.8}{1} \hspace*{2.5ex} \scalebox{0.8}{0} \hspace*{2.4ex} \scalebox{0.8}{0} \hspace*{2.4ex} \scalebox{0.8}{0} \\[-1ex]
    \begin{matrix}
        \scalebox{0.8}{2} \\[0.01ex]
        \scalebox{0.8}{2} \\[0.01ex]
        \scalebox{0.8}{1} \\[0.01ex]
        \scalebox{0.8}{0} \\[0.01ex]
        \scalebox{0.8}{0} \\[0.01ex]
        \scalebox{0.8}{0}
    \end{matrix}
    &\left[\,
    \begin{matrix}[cc|c|ccc]
         0 & 0 & 1 & 4 & 4 & 4 \\
         0 & 0 & 1 & 4 & 4 & 4 \\ \hline
         1 & 1 & 0 & 1 & 1 & 1 \\ \hline
         4 & 4 & 1 & 0 & 0 & 0 \\
         4 & 4 & 1 & 0 & 0 & 0 \\
         4 & 4 & 1 & 0 & 0 & 0
    \end{matrix} \, \right],
\end{align*}
i.e., $ \lambda_H = [4, 1, 0] $ with $ \mu_H = [12, 10, 14] $. This is consistent with $ \rank(H) = 36 - 4 - 1 - 9 = 22 $. \exampleSymbol
\end{example}

\begin{lemma} \label{lem:null space of H}
It holds that $ x^\top H \ts x = 0 $ if and only if $ H \ts x = 0 $.
\end{lemma}

\begin{proof}
First, we write $ H = U_H \ts \Lambda_H U_H^\top $, where all eigenvalues of $ H $ are nonnegative as shown in Theorem~\ref{thm:properties of H}. Assume that
\begin{equation*}
    x^\top H \ts x = x^\top U_H \Lambda_H \ts U_H^\top \ts x = \big\| \Lambda_H^{\nicefrac{1}{2}} U_H^\top x \big\|_2^2 = 0,
\end{equation*}
then $ \Lambda_H^{\nicefrac{1}{2}} \ts U_H^\top x = 0 $ and
\begin{equation*}
    H \ts x = U_H \ts \Lambda_H^{\nicefrac{1}{2}} \ts \Lambda_H^{\nicefrac{1}{2}} \ts U_H^\top \ts x = 0.
\end{equation*}
The other direction is clear.
\end{proof}

We can now interpret the relaxed graph isomorphism problem \eqref{eq:GIDS} as follows: Find a stochastic matrix $ X $ so that its vectorization $ x $ is contained in the null space of $ H $. This also shows that if the graphs have repeated eigenvalues, the search space is much larger since the dimension of the null space increases as shown in Corollary~\ref{cor:rank of H}.

\begin{theorem}
Assume that $ A $ and $ B $ are isospectral and let $ S \in \R^{n \times n} $ be a block diagonal matrix of the form
\begin{equation*}
    S =
    \begin{bmatrix}
        S^{(1)} & \\
        & \ddots \\
        & & S^{(m)}
    \end{bmatrix},
\end{equation*}
where $ S^{(k)} \in \R^{\mu_A^{(k)} \times \mu_A^{(k)}} $. Then
\begin{equation*}
    X^*(S) = U_B \ts S \ts U_A^\top
\end{equation*}
minimizes the cost function $ \norm{X \ts A - B \ts X}_F^2 $ and all minimizers are of this form.
\end{theorem}

\begin{proof}
The optimal vectors $ x $ can be written as linear combinations of the vectors contained in the null space of $ H $ and it holds that $ u_A^{(k,\alpha)} \otimes u_B^{(k,\beta)} = \vec\left(u_B^{(k,\beta)} \boxtimes u_A^{(k,\alpha)} \right) $, where $ \boxtimes $ denotes the tensor product.\!\footnote{The \emph{tensor product} or \emph{outer product} of two vectors $ u, v \in \R^n $ is defined by $ u \boxtimes v = u \ts v^\top $. Note that $ u \boxtimes v \in R^{n \times n} $, whereas $ u \otimes v \in \R^{n^2} $. We have $ u \otimes v = \vec(v \boxtimes u) $.} The corresponding matrices $ X $ can thus be written as linear combinations of the matrices $ u_B^{(k,\beta)} \boxtimes u_A^{(k,\alpha)} $, where $ 1 \le \alpha, \beta \le \mu_{A}^{(k)} $ and $ k = 1, \dots, m $. We then have
\begin{equation*}
    X^*(S) = U_B \ts S \ts U_A^\top = \sum_{k=1}^m \sum_{\alpha=1}^{\mu_{A}^{(k)}} \sum_{\beta=1}^{\mu_{A}^{(k)}} S_{\beta \alpha}^{(k)} \left( u_B^{(k,\beta)} \boxtimes u_A^{(k,\alpha)} \right),
\end{equation*}
which is indeed a linear combination of the matricized vectors spanning the null space. In fact,
\begin{align*}
    X^*(S) \ts A - B \ts X^*(S)
        &= U_B \ts S \ts U_A^\top A - B \ts U_B \ts S \ts U_A^\top \\
        &= \sum_{k=1}^m \sum_{\alpha=1}^{\mu_{A}^{(k)}} \sum_{\beta=1}^{\mu_{A}^{(k)}} S_{\beta \alpha}^{(k)} \left( u_B^{(k,\beta)} \boxtimes u_A^{(k,\alpha)} \right) A - S_{\beta \alpha}^{(k)} \ts B \left( u_B^{(k,\beta)} \boxtimes u_A^{(k,\alpha)} \right) \\
        &= \sum_{k=1}^m \sum_{\alpha=1}^{\mu_{A}^{(k)}} \sum_{\beta=1}^{\mu_{A}^{(k)}} \lambda_A^{(k)} \ts S_{\beta \alpha}^{(k)} \left( u_B^{(k,\beta)} \boxtimes u_A^{(k,\alpha)} \right) - \lambda_B^{(k)} \ts S_{\beta \alpha}^{(k)} \left( u_B^{(k,\beta)} \boxtimes u_A^{(k,\alpha)} \right) \\
        &= 0
\end{align*}
and thus also the norm. Alternatively, we could reuse the proof of Theorem~\ref{thm:Procrustes} (without requiring the matrices to be orthogonal), which then also implies that all minimizers must be of the above form.
\end{proof}

The structure of the solutions is, as expected, similar to the one obtained for the orthogonal relaxation considered in Section~\ref{sec:orthogonal relaxation}. So far, we have not taken the constraint that $ X $ must be doubly stochastic into account. Let
\begin{equation*}
    w_A = \begin{bmatrix} w_A^{(1)} \\  \vdots \\ w_A^{(m)}\end{bmatrix} = U_A^\top \mathds{1}_n,
    \quad \text{and} \quad
    w_B = \begin{bmatrix} w_B^{(1)} \\  \vdots \\ w_B^{(m)}\end{bmatrix} = U_B^\top \mathds{1}_n,
\end{equation*}
where $ w_A^{(k)}, w_B^{(k)} \in \R^{\mu_A^{(k)}} $, then this requirement translates to
\begin{equation} \label{eq:pseudo-stochasticity}
\begin{alignedat}{4}
    X^*(S) \ts \mathds{1}_n &= U_B \ts S \ts U_A^\top \mathds{1}_n &&= \mathds{1}_n && \implies & S \ts w_A &= w_B, \\
    \mathds{1}_n^\top X^*(S) &= \mathds{1}_n^\top U_B \ts S \ts U_A^\top &&= \mathds{1}_n^\top && \implies \; & S^\top w_B &= w_A.
\end{alignedat}
\end{equation}
Condition \eqref{eq:pseudo-stochasticity} is closely related to the definition of friendliness, which requires that all entries of $ w_A $ and $ w_B $ are nonzero and that $ S $ is a diagonal matrix. Additionally, solutions have to be non-negative, i.e.,
\begin{equation} \label{eq:nonnegativity}
    X^*(S) \ge 0.
\end{equation}
Matrices that satisfy \eqref{eq:pseudo-stochasticity} but not \eqref{eq:nonnegativity} are sometimes called \emph{pseudo-stochastic}. Since $ S $ is a block-diagonal matrix, we can decompose \eqref{eq:pseudo-stochasticity} into $ m $ separate equations. This results in
\begin{equation} \label{eq:pseudo-stochasticity block-wise}
    \begin{split}
        \big(S^{(k)}\big)^\top S^{(k)} \ts w_A^{(k)} &= w_A^{(k)}, \\
        S^{(k)} \big(S^{(k)}\big)^\top w_B^{(k)} &= w_B^{(k)}.
    \end{split}
\end{equation}
If $ S^{(k)} $ is an orthogonal matrix, then these equations are trivially satisfied (see also Section~\ref{sec:orthogonal relaxation}), but there are in general additional solutions. This is in particular the case if $ w_A^{(k)} = w_B^{(k)} = 0 $.

\begin{lemma}
Let $ A $ and $ B $ be isospectral and friendly. Then the relaxed problem \eqref{eq:GIDS} has at most one solution with $ c_\mathcal{D}^{} = 0 $.
\end{lemma}

\begin{proof}
Since $ A $ and $ B $ are by definition friendly, all eigenvalues are distinct, which implies that $ S^{(k)} \in \R $ for all $ k $, and all $ w_A^{(k)} $ and $ w_B^{(k)} $ are non-zero. Thus, \eqref{eq:pseudo-stochasticity block-wise} can only be satisfied if $ S^{(k)} = \pm 1 $ and the signs of $ S^{(k)} $ are uniquely determined by \eqref{eq:pseudo-stochasticity}. The equations have no solution if $ \big| w_A^{(k)} \big| \ne \big| w_B^{(k)} \big| $ for at least one $ k $.
\end{proof}

If the unique optimal solution is a permutation matrix, then $ \mc[A]{G} \cong \mc[B]{G} $. A similar result was also shown in \cite{ABK15}. Geometrically, this means that the affine subspace given by the constraints $ C \ts x = d $ and the null space of $ H $ intersect in exactly one point $ x \ge 0 $.

If $ A $ and $ B $ are isospectral but not friendly, then the intersection of the affine subspace defined by $ C \ts x = d $ and the null space of $ H $ might again form an affine subspace. Let $ x_s $ be one solution of the augmented inhomogeneous system of linear equations
$
    \left[
    \begin{smallmatrix}
        C \\
        H
    \end{smallmatrix}
    \right]
    x
    =
    \left[
    \begin{smallmatrix}
        d \\
        0
    \end{smallmatrix}
    \right],
$
then all vectors of the form $ x = x_s + x_h $, where $ x_h $ is contained in the null space of the augmented matrix, are pseudo-stochastic solutions. We have
\begin{equation*}
    \rank\left(
        \begin{bmatrix}
            C \\
            H
        \end{bmatrix}
        \right) \le \min\left(n^2 - \sum_{k=1}^m \big(\mu_A^{(k)} \big)^2 + 2 \ts n - 1, n^2\right)
\end{equation*}
since the dimension of the intersection of the two (affine) subspaces must be lower than the sum of their dimensions.

\begin{example} \label{ex:square graphs}

\begin{figure}
    \centering
    \begin{minipage}{0.2\linewidth}
        \centering
        \subfiguretitle{(a)}
        \vspace*{1ex}
        \begin{tikzpicture}[
                semithick 
            ]
            \tikzstyle{every state}=[
                draw=black,
                thick,
                fill=white,
                inner sep=0pt,
                text width=6mm,
                align=center,
                scale=0.6
            ]

            \node[state] (v1) {1};
            \node[state] (v2) [right=1cm of v1] {2};
            \node[state] (v3) [below=1cm of v2] {3};
            \node[state] (v4) [below=1cm of v1] {4};

            \path[-] (v1) edge node[above=0.1ex] {\footnotesize 1} (v2);
            \path[-] (v2) edge node[right=0.2ex] {\footnotesize 2} (v3);
            \path[-] (v3) edge node[below=0.1ex] {\footnotesize 3} (v4);
            \path[-] (v4) edge node[left=0.2ex] {\footnotesize 4} (v1);
        \end{tikzpicture} \\[1ex]
        \begin{tikzpicture}[
                semithick 
            ]
            \tikzstyle{every state}=[
                draw=black,
                thick,
                fill=white,
                inner sep=0pt,
                text width=6mm,
                align=center,
                scale=0.6
            ]

            \node[state] (v1) {2};
            \node[state] (v2) [right=1cm of v1] {4};
            \node[state] (v3) [below=1cm of v2] {1};
            \node[state] (v4) [below=1cm of v1] {3};

            \path[-] (v1) edge node[above=0.1ex] {\footnotesize 1} (v2);
            \path[-] (v2) edge node[right=0.2ex] {\footnotesize 2} (v3);
            \path[-] (v3) edge node[below=0.1ex] {\footnotesize 3} (v4);
            \path[-] (v4) edge node[left=0.2ex] {\footnotesize 4} (v1);
        \end{tikzpicture}
    \end{minipage}
    \begin{minipage}{0.2\linewidth}
        \centering
        \subfiguretitle{(b)}
        \vspace*{1ex}
        \begin{tikzpicture}[
                semithick 
            ]
            \tikzstyle{every state}=[
                draw=black,
                thick,
                fill=white,
                inner sep=0pt,
                text width=6mm,
                align=center,
                scale=0.6
            ]

            \node[state] (v1) {1};
            \node[state] (v2) [right=1cm of v1] {2};
            \node[state] (v3) [below=1cm of v2] {3};
            \node[state] (v4) [below=1cm of v1] {4};

            \path[-] (v1) edge node[above=0.1ex] {\footnotesize 1} (v2);
            \path[-] (v2) edge node[right=0.2ex] {\footnotesize 2} (v3);
            \path[-] (v3) edge node[below=0.1ex] {\footnotesize 2} (v4);
            \path[-] (v4) edge node[left=0.2ex] {\footnotesize 2} (v1);

            \draw[red, dotted, line width=0.75pt] (0.8, -0.02) -- (0.8, -1.6);
        \end{tikzpicture} \\[1ex]
        \begin{tikzpicture}[
                semithick 
            ]
            \tikzstyle{every state}=[
                draw=black,
                thick,
                fill=white,
                inner sep=0pt,
                text width=6mm,
                align=center,
                scale=0.6
            ]

            \node[state] (v1) {2};
            \node[state] (v2) [right=1cm of v1] {4};
            \node[state] (v3) [below=1cm of v2] {1};
            \node[state] (v4) [below=1cm of v1] {3};

            \path[-] (v1) edge node[above=0.1ex] {\footnotesize 1} (v2);
            \path[-] (v2) edge node[right=0.2ex] {\footnotesize 2} (v3);
            \path[-] (v3) edge node[below=0.1ex] {\footnotesize 2} (v4);
            \path[-] (v4) edge node[left=0.2ex] {\footnotesize 2} (v1);

            \draw[red, dotted, line width=0.75pt] (0.8, -0.02) -- (0.8, -1.6);
        \end{tikzpicture}
    \end{minipage}
    \begin{minipage}{0.2\linewidth}
        \centering
        \subfiguretitle{(c)}
        \vspace*{1ex}
        \begin{tikzpicture}[
                semithick 
            ]
            \tikzstyle{every state}=[
                draw=black,
                thick,
                fill=white,
                inner sep=0pt,
                text width=6mm,
                align=center,
                scale=0.6
            ]

            \node[state] (v1) {1};
            \node[state] (v2) [right=1cm of v1] {2};
            \node[state] (v3) [below=1cm of v2] {3};
            \node[state] (v4) [below=1cm of v1] {4};

            \path[-] (v1) edge node[above=0.1ex] {\footnotesize 1} (v2);
            \path[-] (v2) edge node[right=0.2ex] {\footnotesize 2} (v3);
            \path[-] (v3) edge node[below=0.1ex] {\footnotesize 1} (v4);
            \path[-] (v4) edge node[left=0.2ex] {\footnotesize 2} (v1);

            \draw[red, dotted, line width=0.75pt] (0.8, -0.02) -- (0.8, -1.6);
            \draw[red, dotted, line width=0.75pt] (0.01, -0.8) -- (1.63, -0.8);
        \end{tikzpicture} \\[1ex]
        \begin{tikzpicture}[
                semithick 
            ]
            \tikzstyle{every state}=[
                draw=black,
                thick,
                fill=white,
                inner sep=0pt,
                text width=6mm,
                align=center,
                scale=0.6
            ]

            \node[state] (v1) {2};
            \node[state] (v2) [right=1cm of v1] {4};
            \node[state] (v3) [below=1cm of v2] {1};
            \node[state] (v4) [below=1cm of v1] {3};

            \path[-] (v1) edge node[above=0.1ex] {\footnotesize 1} (v2);
            \path[-] (v2) edge node[right=0.2ex] {\footnotesize 2} (v3);
            \path[-] (v3) edge node[below=0.1ex] {\footnotesize 1} (v4);
            \path[-] (v4) edge node[left=0.2ex] {\footnotesize 2} (v1);

            \draw[red, dotted, line width=0.75pt] (0.8, -0.02) -- (0.8, -1.6);
            \draw[red, dotted, line width=0.75pt] (0.01, -0.8) -- (1.63, -0.8);
        \end{tikzpicture}
    \end{minipage}
    \begin{minipage}{0.2\linewidth}
        \centering
        \subfiguretitle{(d)}
        \vspace*{1ex}
        \begin{tikzpicture}[
                semithick 
            ]
            \tikzstyle{every state}=[
                draw=black,
                thick,
                fill=white,
                inner sep=0pt,
                text width=6mm,
                align=center,
                scale=0.6
            ]

            \node[state] (v1) {1};
            \node[state] (v2) [right=1cm of v1] {2};
            \node[state] (v3) [below=1cm of v2] {3};
            \node[state] (v4) [below=1cm of v1] {4};

            \path[-] (v1) edge node[above=0.1ex] {\footnotesize 1} (v2);
            \path[-] (v2) edge node[right=0.2ex] {\footnotesize 1} (v3);
            \path[-] (v3) edge node[below=0.1ex] {\footnotesize 1} (v4);
            \path[-] (v4) edge node[left=0.2ex] {\footnotesize 1} (v1);

            \draw[red, dotted, line width=0.75pt] (0.8, -0.02) -- (0.8, -1.6);
            \draw[red, dotted, line width=0.75pt] (0.01, -0.8) -- (1.63, -0.8);
            \draw[red, dotted, line width=0.75pt] (0.2, -0.25) -- (1.41, -1.37);
            \draw[red, dotted, line width=0.75pt] (0.2, -1.37) -- (1.41, -0.25);
        \end{tikzpicture} \\[1ex]
        \begin{tikzpicture}[
                semithick 
            ]
            \tikzstyle{every state}=[
                draw=black,
                thick,
                fill=white,
                inner sep=0pt,
                text width=6mm,
                align=center,
                scale=0.6
            ]

            \node[state] (v1) {2};
            \node[state] (v2) [right=1cm of v1] {4};
            \node[state] (v3) [below=1cm of v2] {1};
            \node[state] (v4) [below=1cm of v1] {3};

            \path[-] (v1) edge node[above=0.1ex] {\footnotesize 1} (v2);
            \path[-] (v2) edge node[right=0.2ex] {\footnotesize 1} (v3);
            \path[-] (v3) edge node[below=0.1ex] {\footnotesize 1} (v4);
            \path[-] (v4) edge node[left=0.2ex] {\footnotesize 1} (v1);

            \draw[red, dotted, line width=0.75pt] (0.8, -0.02) -- (0.8, -1.6);
            \draw[red, dotted, line width=0.75pt] (0.01, -0.8) -- (1.63, -0.8);
            \draw[red, dotted, line width=0.75pt] (0.2, -0.25) -- (1.41, -1.37);
            \draw[red, dotted, line width=0.75pt] (0.2, -1.37) -- (1.41, -0.25);
        \end{tikzpicture}
    \end{minipage}
    \caption{Isomorphic square graphs $ \mc[A]{G} $ (top) and $ \mc[B]{G} $ (bottom) with different edge weights and symmetries. The dotted red lines represent reflection symmetries.}
    \label{fig:square graphs}
\end{figure}
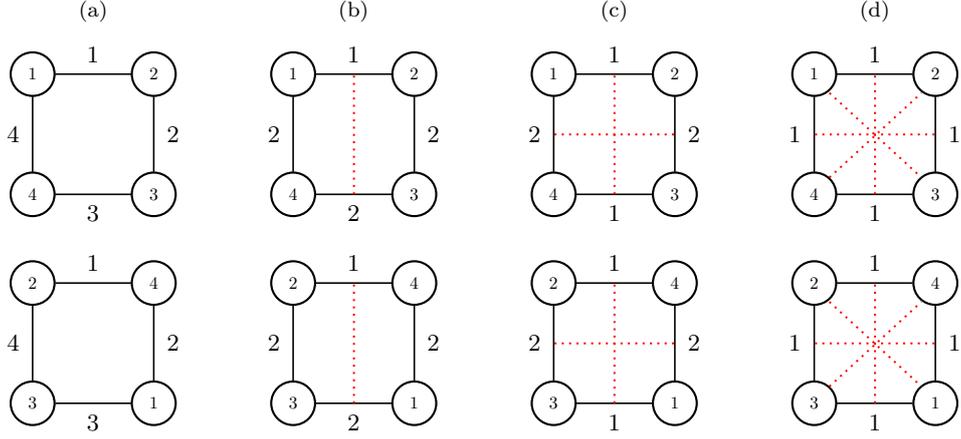

Let us consider the square graphs shown in Figure~\ref{fig:square graphs}. Depending on the edge weights, we obtain different types of symmetries, which are also reflected in the spectral properties of the graphs.
\begin{enumerate}[leftmargin=3.5ex, itemsep=0ex, topsep=0.5ex, label=(\alph*)]
\item There are no symmetries and the automorphism group is trivial. All eigenvalues are distinct and all eigenvectors are friendly. We can easily determine the signs of the diagonal entries of~$ S $. Assume, for instance, that $ w_A = [0.1, 0.3, -0.5, -1.9]^\top $ and $ w_B = [-0.1, 0.3, 0.5, -1.9]^\top $, then $ S = \diag(-1, 1, -1, 1) $ using \eqref{eq:pseudo-stochasticity block-wise}. The null space of the augmented matrix contains only the zero vector and the solution is unique.
\item There are two automorphisms due to the reflection symmetry. The eigenvalues are distinct, but now only the first and third eigenvectors are friendly. Assume that $ w_A = [2, 0, 0.2, 0]^\top $ and $ w_B = [-2, 0, 0.2, 0]^\top $, then we have $ S = \diag(-1, \pm 1, 1, \pm 1) $. However, only two combinations result in permutation matrices, the other two satisfy \eqref{eq:pseudo-stochasticity} but contain negative entries.
\item There now exist four automorphisms, but the eigenvalues are still distinct. Only the first eigenvector is friendly. Four of the eight possible combinations for $ S $ result in permutations.
\item Due to the $ D_4 $ symmetry, there exist eight automorphisms. The eigenvalue $ \lambda_A^{(2)} = 0 $ has multiplicity $ \mu_A^{(2)} = 2 $. Now we cannot simply determine the sign structure of the matrix $ S $ anymore since the block $ S^{(2)} $ is a $ 2 \times 2 $ matrix. \exampleSymbol
\end{enumerate}
\end{example}

This illustrates again how symmetries and repeated eigenvalues complicate graph isomorphism testing since the solution space increases.

\begin{example}
We have seen in Example~\ref{ex:Frucht graph feasbible solutions} that there exists an 11-dimensional affine subspace of solutions due to the 11 unfriendly eigenvectors. However, only one of the eigenvectors of the Frucht graphs is ambiguous so that we can still determine all but one of the signs of the diagonal matrix~$ S $. In the remaining one-dimensional space of feasible solutions, there exists just one non-negative matrix, which is the permutation matrix that solves the graph isomorphism problem. \exampleSymbol
\end{example}

By exploiting additional properties of the eigenvectors, we can in this case eliminate spurious solutions and obtain again a unique solution.

\begin{example} \label{ex:convex combinations}
If we solve the relaxed optimization problem defined in Lemma~\ref{lem:quadratic form} for the square graphs shown in Figure~\ref{fig:square graphs}\ts(b), then we obtain, for example,
\begin{equation*}
    X^*(S) =
    \begin{bmatrix}
        0 & 0 & \frac{1}{2} & \frac{1}{2} \\
        \frac{1}{2} & \frac{1}{2} & 0 & 0 \\
        0 & 0 & \frac{1}{2} & \frac{1}{2} \\
        \frac{1}{2} & \frac{1}{2} & 0 & 0
    \end{bmatrix}
    = \tfrac{1}{2}
    \underbrace{
    \begin{bmatrix}
         0 & 0 & 1 & 0 \\
         1 & 0 & 0 & 0 \\
         0 & 0 & 0 & 1 \\
         0 & 1 & 0 & 0
     \end{bmatrix}
     }_{=: P_1}
     + \tfrac{1}{2}
    \underbrace{
    \begin{bmatrix}
         0 & 0 & 0 & 1 \\
         0 & 1 & 0 & 0 \\
         0 & 0 & 1 & 0 \\
         1 & 0 & 0 & 0
     \end{bmatrix}
     }_{=: P_2},
\end{equation*}
which is indeed a convex combination of the optimal permutation matrices $ P_1 $ and $ P_2 $, see Lemma~\ref{lem:convex combinations}. Any other convex combination of the two permutation matrices is also an optimal solution of the relaxed problem. The matrix $ S $ is then in general not an orthogonal matrix, which it must be according to Section~\ref{sec:orthogonal relaxation}. \exampleSymbol
\end{example}

\subsection{Penalizing non-binary matrices}

The question now is how we can penalize doubly stochastic matrices that are not permutation matrices. We use Lemma~\ref{lem:null space of H} to turn the cost function into a constraint and optimize a different function instead.

\begin{lemma} \label{lem:penalized problem}
The optimization problem \eqref{eq:GIDS} penalizing non-binary matrices can be written as
\begin{equation*}
    \min_{\substack{x \ge 0 \\[0.4ex] C x = d \\[0.4ex] H x = 0}} -x^\top x,
\end{equation*}
where the first two constraints enforce the stochasticity of $ X $ and the third constraint guarantees that $ X \ts A = B \ts X $.
\end{lemma}

\begin{proof}
This is a reformulation of the optimization problem from Lemma~\ref{lem:quadratic form}. Due to Lemma~\ref{lem:null space of H}, we know that $ H \ts x = 0 $ implies $ x^\top H \ts x = 0 $, which in turn implies that $ \norm{X \ts A - B \ts X} = 0 $. For any doubly stochastic matrix $ X $, we have
\begin{equation*}
    x^\top x = \tr\big(X^\top X\big) = \sum_{i=1}^n \sum_{j=1}^n X_{ij}^2 \le \sum_{i=1}^n \sum_{j=1}^n X_{ij} = n
\end{equation*}
and equality holds if and only if $ X \in \mathcal{P}(n) $. Furthermore,
\begin{equation*}
    x^\top x = \tr\big(X^\top X\big) = \sum_{i=1}^n \norm{X_i}_2^2 \ge n \frac{1}{n} = 1
\end{equation*}
using the Cauchy--Schwarz inequality\footnote{Note that $ 1 = \norm{X_i}_1 = \innerprod{\mathds{1}_n}{X_i} \le \norm{\mathds{1}_n}_2 \norm{X_i}_2 = \sqrt{n} \norm{X_i}_2 $.} for the columns of the matrix $ X $ and equality holds if and only if $ X = \frac{1}{n} \mathds{1} \mathds{1}^\top $. That is, $ 1 \le x^\top x \le n $ and minimizing $ -x^\top x $ will force the solver to favor permutation matrices.
\end{proof}

A similar penalty term was used in the convex-to-concave projection approach in \cite{Dym18}. Since $ u_A^{(1)} = \mathds{1}_n $ is an eigenvector of any regular graph $ \mc[A]{G} $, the matrix $ X = \frac{1}{n} \mathds{1} \mathds{1}^\top $ is of the form $ X^*(S) $, where only the entry $ S_{11} $ is nonzero. This solution does not contain any useful information about potential isomorphisms.

\begin{remark}
As $ X^*(S) = U_B \ts S \ts U_A^\top $, this yields
\begin{align*}
    \tr\Big(X^*(S)^\top X^*(S)\Big)
        = \tr\Big( U_A \ts S^\top U_B^\top U_B \ts S \ts U_A^\top \Big)
        = \tr\Big( S^\top S \Big)
\end{align*}
using the cyclic property of the trace and the orthogonality of the matrices $ U_A $ and $ U_B $. The trace is $ n $ if and only if $ S^\top S = I_n $, i.e., if $ S $ is an orthogonal matrix. This is consistent with the derivations in Section~\ref{sec:orthogonal relaxation}. We could now also work directly in the space of $ S $ matrices by rewriting the optimization problem defined in Lemma~\ref{lem:penalized problem} as
\begin{equation*}
    \min_{\substack{N \sigma \ge 0 \\[0.4ex] C N \sigma = d}} - \sigma^\top \sigma,
\end{equation*}
with $ N = \big[U_A^{(1)} \otimes U_B^{(1)}, \dots, U_A^{(m)} \otimes U_B^{(m)}\big] $ and $ x = N \ts \sigma $. That is, the columns of $ N $ form a basis of the null space of $ H $, constructed from the normalized eigenvectors, and $ \sigma $ is a reduced representation of the block matrix $ S $, containing only the block diagonal entries.
\end{remark}

\subsection{Frank--Wolfe algorithm}

In order to solve the optimization problem defined in Lemma~\ref{lem:penalized problem}, we use the Frank--Wolfe algorithm \cite{FW56, Jaggi13}, which for a constrained convex optimization problem of the form
\begin{equation*}
    \min_{x \in \mathcal{C}} f(x)
\end{equation*}
defined on a convex set $ \mathcal{C} $ consists of the following two steps:
\begin{enumerate}[leftmargin=3.5ex, itemsep=0ex, topsep=0.5ex]
\item Given a feasible initial condition $ x^{(k)} \in \mathcal{C} $, compute an optimal solution $ y^* $ of the linear problem
\begin{equation*}
    \min_{y \in \mathcal{C}} \big\langle y, \nabla f(x^{(k)}) \big\rangle.
\end{equation*}
\item Define $ x^{(k+1)} $ to be a convex combination of $ x^{(k)} $ and $ y^* $, i.e.,
\begin{equation*}
    x^{(k+1)} = (1 - \gamma) \ts x^{(k)} + \gamma \ts y^*,
\end{equation*}
where $ 0 \le \gamma \le 1 $ is chosen so that $ f\big(x^{(k+1)}\big) $ is minimized.
\end{enumerate}
The cost function $ f(x) = -x^\top x $ we seek to minimize is not convex but strictly concave and the global minima are permutation matrices. Since we are just interested in showing that at least one permutation matrix exists that solves the graph isomorphism problem, the non-uniqueness of the optimal global solution is not a problem. The gradient of the cost function is simply $ \nabla f(x) = -2 \ts x $ and, ignoring the constant factor 2, in each iteration of the Frank--Wolfe algorithm, we solve
\begin{equation*}
    \min_{\substack{y \ge 0 \\[0.4ex] C y = d \\[0.4ex] H y = 0}} -\big\langle y, x^{(k)} \big\rangle
\end{equation*}
to obtain a new solution $ x^{(k+1)} $ as described above, where $ x^{(0)} $ is the solution of the relaxed problem defined in Lemma~\ref{lem:quadratic form} (if no solutions exist, then $ \mc[A]{G} \not\cong \mc[B]{G} $). Since $ x^{(k)} $ and $ y^* $ are both feasible solutions, the convex combination $ x^{(k+1)} $ also satisfies the constraints. That is, the Frank--Wolfe algorithm is guaranteed to compute a doubly stochastic matrix whose vectorization is contained in the null space of $ H $.

\begin{example} \label{ex:Frank-Wolfe convergence}

\begin{figure}
    \centering
    \begin{minipage}[t]{0.35\linewidth}
        \centering
        \subfiguretitle{(a)}
        \vspace*{2.5ex}
        \includegraphics[width=0.95\linewidth]{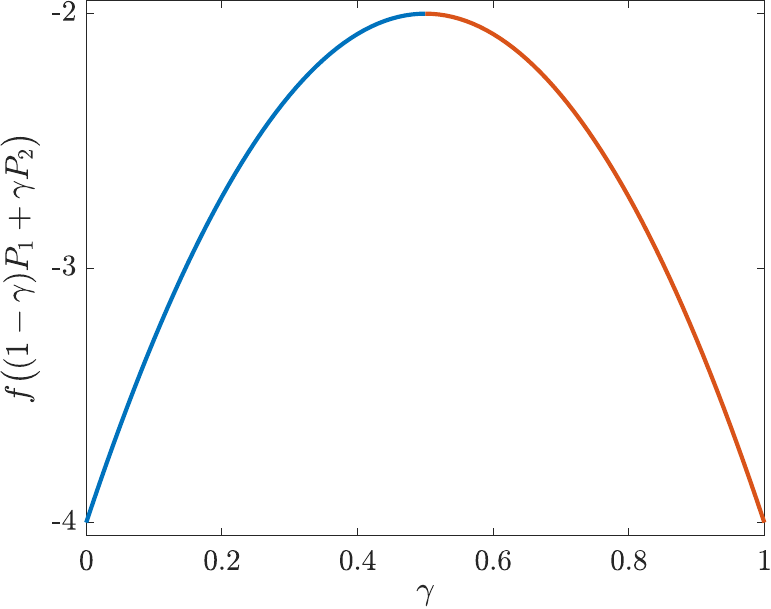}
    \end{minipage}
    \begin{minipage}[t]{0.30\linewidth}
        \centering
        \subfiguretitle{(b)}
        \begin{tikzpicture}
            \node[anchor=south west,inner sep=0] (image) at (0,0) {\includegraphics[width=0.85\linewidth]{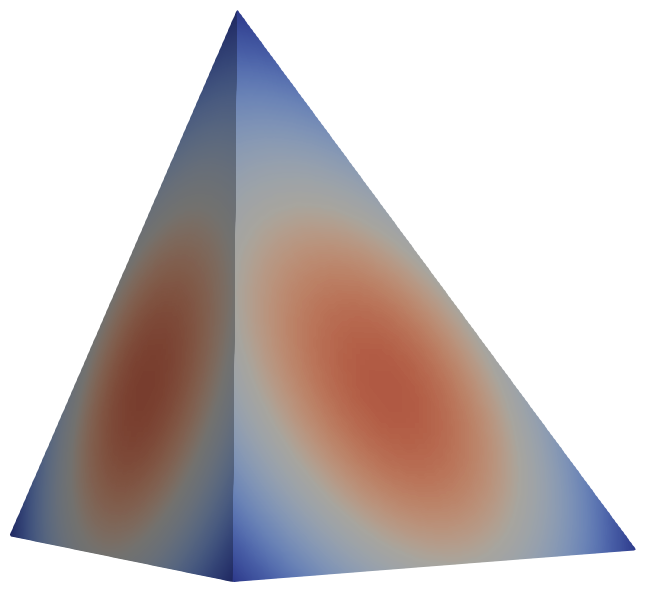}};
            \begin{scope}[x={(image.south east)},y={(image.north west)}]
                \draw (-0.05,0.05) node{$ P_3 $}
                      (0.35,-0.07) node{$ P_4 $}
                      (1.05,0.05) node{$ P_1 $}
                      (0.38,1.05) node{$ P_2 $};
            \end{scope}
        \end{tikzpicture}
    \end{minipage}
    \hspace{0.5ex}
    \begin{minipage}[t]{0.3\linewidth}
        \centering
        \subfiguretitle{(c)}
        \begin{tikzpicture}
            \node[anchor=south west,inner sep=0] (image) at (0,0) {\includegraphics[width=0.85\linewidth]{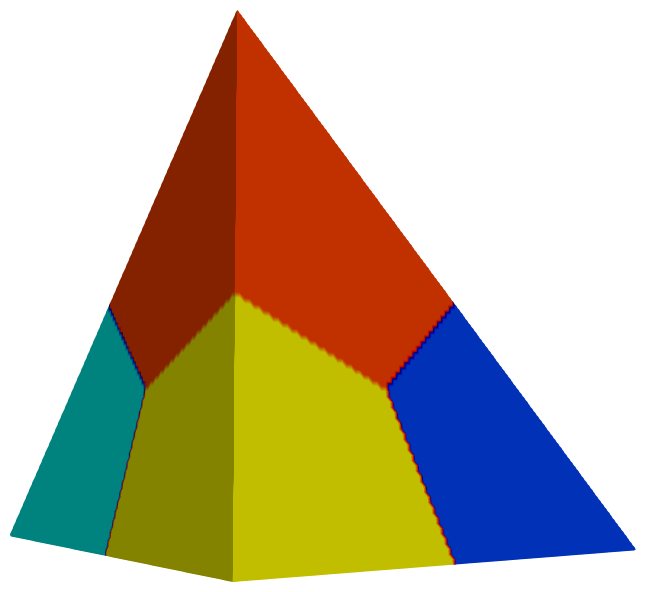}};
            \begin{scope}[x={(image.south east)},y={(image.north west)}]
                \draw (-0.05,0.05) node{$ P_3 $}
                      (0.35,-0.07) node{$ P_4 $}
                      (1.05,0.05) node{$ P_1 $}
                      (0.38,1.05) node{$ P_2 $};
            \end{scope}
        \end{tikzpicture}
    \end{minipage}
    \caption{(a) Cost function for the graphs in Figure~\ref{fig:square graphs}\ts(b) for different initial conditions. The two minima are the permutation matrices $ P_1 $ and $ P_2 $. The solver always converges to $ P_1 $ (blue line) or $ P_2 $ (red line) after just one iteration. (b) Cost function for the graphs in Figure~\ref{fig:square graphs}\ts(c) for different convex combinations of the four optimal permutation matrices. The tetrahedron represents the parametrization of the convex set of solutions. The vertices are the permutation matrices and the minima of the cost function. All values are between $ -4 $ (blue) and $ -1 $ (red). (c) Convergence of the Frank--Wolfe algorithm to one of the four optimal solutions.}
    \label{fig:Frank-Wolfe convergence}
\end{figure}

For illustration purposes, we consider again the simple square graphs shown in Figure~\ref{fig:square graphs}\ts(b). In Example~\ref{ex:convex combinations}, we have seen that any convex combination of the optimal permutation matrices is also an optimal solution of the relaxed problem. Depending on the convex combination, the Frank--Wolfe algorithm converges to different solutions as shown in Figure~\ref{fig:Frank-Wolfe convergence}\ts(a). Similarly, for the graphs shown in Figure~\ref{fig:square graphs}\ts(c), any convex combination of the four optimal permutation matrices is an optimal solution of the relaxed problem. Depending on the initial condition, the Frank--Wolfe algorithm again converges after just one iteration to one of the four permutation matrices. This is shown in Figures~\ref{fig:Frank-Wolfe convergence}\ts(b)--(c). \exampleSymbol
\end{example}

The results illustrate that the new cost function penalizes doubly stochastic matrices that are not permutation matrices and forces the Frank--Wolfe solver into the vertices of the convex set.

\begin{example} \label{ex:Petersen graph}
Let us analyze the two graphs shown in Figure~\ref{fig:Petersen graph}. We have $ \lambda_A = [3, 1, -2] = \lambda_B $ with multiplicities $ \mu_A = [1, 5, 4] = \mu_B $ and $ \rank(H) = 58 $ by Corollary~\ref{cor:rank of H}. The solution of the optimization problem defined in Lemma~\ref{lem:quadratic form} is simply $ X^{(0)} = \frac{1}{n} \mathds{1}_n \mathds{1}_n^\top $, i.e., all entries are the same and we cannot extract any information about potential isomorphisms yet. However, after just one iteration of the Frank--Wolfe algorithm, we obtain the permutation
\begin{equation*}
    \pi =
    \begin{bmatrix}
        1 &  2 & 3 & 4 & 5 & 6 & 7 & 8 & 9 & 10 \\
        5 & 10 & 8 & 3 & 7 & 9 & 4 & 6 & 2 & 1
    \end{bmatrix},
\end{equation*}
which is one of the 120 valid solutions and proves that the graphs are isomorphic. \exampleSymbol
\end{example}

Even for highly symmetric graphs the Frank--Wolfe algorithm is able to find one of the global minima of the cost function. In general, the solver will not converge after just one iteration. For the Frucht graph, for instance, see Example~\ref{ex:Frucht graph feasbible solutions}, the algorithm converges after two iterations, although the solution is in this case uniquely defined.

\section{Numerical results}

In this section, we will present numerical results for two challenging highly symmetric graphs.

\subsection{Paley graph}

\begin{figure}
    \centering
    \begin{minipage}[t]{0.25\linewidth}
        \centering
        \subfiguretitle{(b) $ f\big(X^{(0)}\big) = -1 $}
        \vspace*{0.5ex}
        \includegraphics[width=0.9\linewidth]{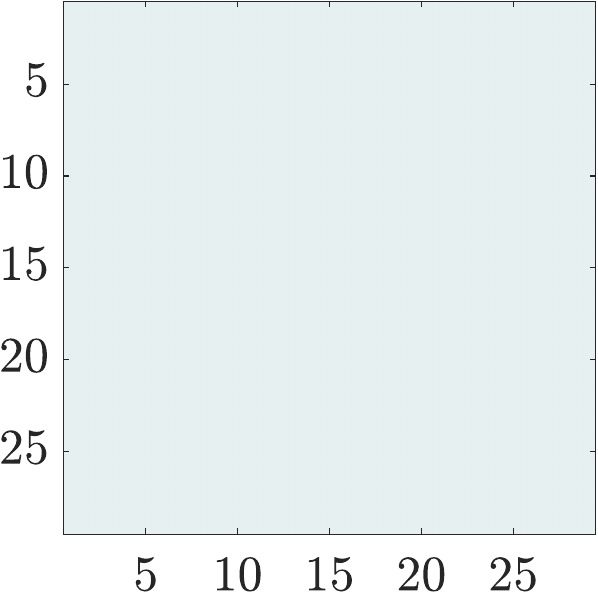} \\[1ex]
        \subfiguretitle{(c) $ f\big(X^{(1)}\big) = -3.17 $}
        \vspace*{0.5ex}
        \includegraphics[width=0.9\linewidth]{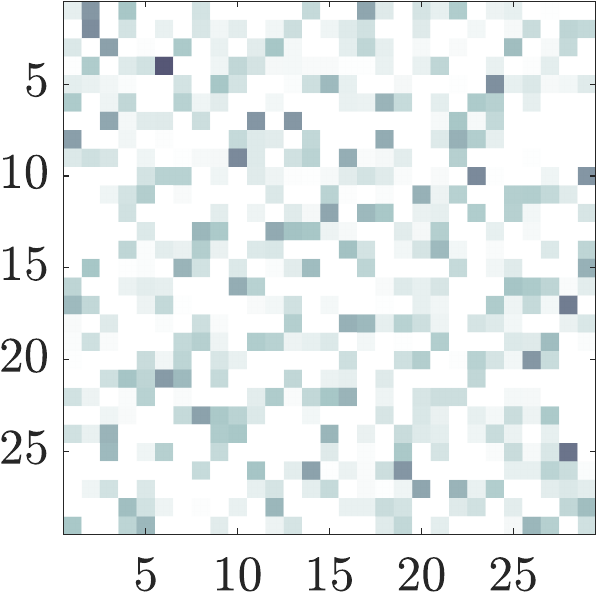}
    \end{minipage}
    \begin{minipage}[t]{0.45\linewidth}
        \centering
        \subfiguretitle{(a)}
        \vspace*{2ex}
        \resizebox{1\textwidth}{!}{%
        \begin{tikzpicture}[
            dot/.style={draw,circle,minimum size=3mm,inner sep=0pt,outer sep=0pt,fill=red}
            ]

            \foreach \i in {1,...,29}
                \coordinate [dot] (\i) at ({10*sin(360*\i/29)}, {10*cos(360*\i/29)});
            \path[-,line width=0.05mm] (1) edge (2)
                     (1) edge (4)
                     (1) edge (6)
                     (1) edge (8)
                     (1) edge (10)
                     (1) edge (12)
                     (1) edge (14)
                     (1) edge (16)
                     (1) edge (18)
                     (1) edge (20)
                     (1) edge (22)
                     (1) edge (24)
                     (1) edge (26)
                     (1) edge (28)
                     (2) edge (3)
                     (2) edge (5)
                     (2) edge (8)
                     (2) edge (10)
                     (2) edge (14)
                     (2) edge (15)
                     (2) edge (18)
                     (2) edge (19)
                     (2) edge (21)
                     (2) edge (22)
                     (2) edge (23)
                     (2) edge (24)
                     (2) edge (25)
                     (3) edge (5)
                     (3) edge (7)
                     (3) edge (8)
                     (3) edge (10)
                     (3) edge (12)
                     (3) edge (13)
                     (3) edge (14)
                     (3) edge (17)
                     (3) edge (18)
                     (3) edge (21)
                     (3) edge (27)
                     (3) edge (28)
                     (3) edge (29)
                     (4) edge (5)
                     (4) edge (7)
                     (4) edge (10)
                     (4) edge (12)
                     (4) edge (16)
                     (4) edge (17)
                     (4) edge (20)
                     (4) edge (21)
                     (4) edge (23)
                     (4) edge (24)
                     (4) edge (25)
                     (4) edge (26)
                     (4) edge (27)
                     (5) edge (7)
                     (5) edge (9)
                     (5) edge (10)
                     (5) edge (12)
                     (5) edge (14)
                     (5) edge (15)
                     (5) edge (16)
                     (5) edge (19)
                     (5) edge (20)
                     (5) edge (23)
                     (5) edge (29)
                     (6) edge (7)
                     (6) edge (9)
                     (6) edge (12)
                     (6) edge (14)
                     (6) edge (18)
                     (6) edge (19)
                     (6) edge (22)
                     (6) edge (23)
                     (6) edge (25)
                     (6) edge (26)
                     (6) edge (27)
                     (6) edge (28)
                     (6) edge (29)
                     (7) edge (9)
                     (7) edge (11)
                     (7) edge (12)
                     (7) edge (14)
                     (7) edge (16)
                     (7) edge (17)
                     (7) edge (18)
                     (7) edge (21)
                     (7) edge (22)
                     (7) edge (25)
                     (8) edge (9)
                     (8) edge (11)
                     (8) edge (14)
                     (8) edge (16)
                     (8) edge (20)
                     (8) edge (21)
                     (8) edge (24)
                     (8) edge (25)
                     (8) edge (27)
                     (8) edge (28)
                     (8) edge (29)
                     (9) edge (11)
                     (9) edge (13)
                     (9) edge (14)
                     (9) edge (16)
                     (9) edge (18)
                     (9) edge (19)
                     (9) edge (20)
                     (9) edge (23)
                     (9) edge (24)
                     (9) edge (27)
                     (10) edge (11)
                     (10) edge (13)
                     (10) edge (16)
                     (10) edge (18)
                     (10) edge (22)
                     (10) edge (23)
                     (10) edge (26)
                     (10) edge (27)
                     (10) edge (29)
                     (11) edge (13)
                     (11) edge (15)
                     (11) edge (16)
                     (11) edge (18)
                     (11) edge (20)
                     (11) edge (21)
                     (11) edge (22)
                     (11) edge (25)
                     (11) edge (26)
                     (11) edge (29)
                     (12) edge (13)
                     (12) edge (15)
                     (12) edge (18)
                     (12) edge (20)
                     (12) edge (24)
                     (12) edge (25)
                     (12) edge (28)
                     (12) edge (29)
                     (13) edge (15)
                     (13) edge (17)
                     (13) edge (18)
                     (13) edge (20)
                     (13) edge (22)
                     (13) edge (23)
                     (13) edge (24)
                     (13) edge (27)
                     (13) edge (28)
                     (14) edge (15)
                     (14) edge (17)
                     (14) edge (20)
                     (14) edge (22)
                     (14) edge (26)
                     (14) edge (27)
                     (15) edge (17)
                     (15) edge (19)
                     (15) edge (20)
                     (15) edge (22)
                     (15) edge (24)
                     (15) edge (25)
                     (15) edge (26)
                     (15) edge (29)
                     (16) edge (17)
                     (16) edge (19)
                     (16) edge (22)
                     (16) edge (24)
                     (16) edge (28)
                     (16) edge (29)
                     (17) edge (19)
                     (17) edge (21)
                     (17) edge (22)
                     (17) edge (24)
                     (17) edge (26)
                     (17) edge (27)
                     (17) edge (28)
                     (18) edge (19)
                     (18) edge (21)
                     (18) edge (24)
                     (18) edge (26)
                     (19) edge (21)
                     (19) edge (23)
                     (19) edge (24)
                     (19) edge (26)
                     (19) edge (28)
                     (19) edge (29)
                     (20) edge (21)
                     (20) edge (23)
                     (20) edge (26)
                     (20) edge (28)
                     (21) edge (23)
                     (21) edge (25)
                     (21) edge (26)
                     (21) edge (28)
                     (22) edge (23)
                     (22) edge (25)
                     (22) edge (28)
                     (23) edge (25)
                     (23) edge (27)
                     (23) edge (28)
                     (24) edge (25)
                     (24) edge (27)
                     (25) edge (27)
                     (25) edge (29)
                     (26) edge (27)
                     (26) edge (29)
                     (27) edge (29)
                     (28) edge (29);
        \end{tikzpicture}}
    \end{minipage}
    \begin{minipage}[t]{0.25\linewidth}
        \centering
        \subfiguretitle{(d) $ f\big(X^{(2)}\big) = -3.61 $}
        \vspace*{0.5ex}
        \includegraphics[width=0.9\linewidth]{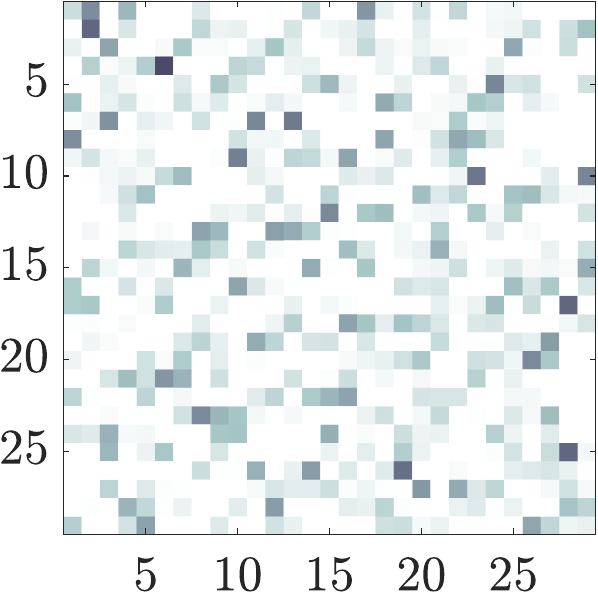} \\[1ex]
        \subfiguretitle{(e) $ f\big(X^{(3)}\big) = -29 $}
        \vspace*{0.5ex}
        \includegraphics[width=0.9\linewidth]{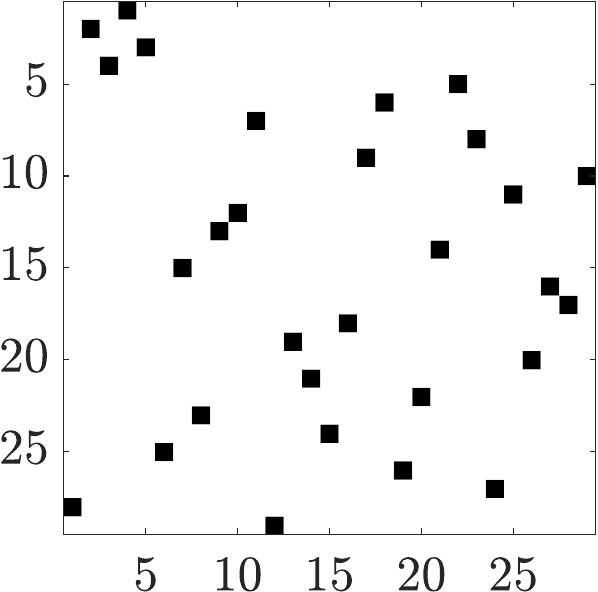}
    \end{minipage}
    \caption{(a) Strongly regular Paley graph $ \mc[A]{G} $. The graph $ \mc[B]{G} $ is a random permutation of $ \mc[A]{G} $. (b)--(e) Convergence of the Frank--Wolfe solver to a permutation matrix. The initial matrix $ X^{(0)} $ is constant. The cost function first decreases only slowly, but then quickly converges to the optimal value. The matrix $ X^{(3)} $ is a valid solution of the graph isomorphism problem.}
    \label{fig:Paley graph}
\end{figure}
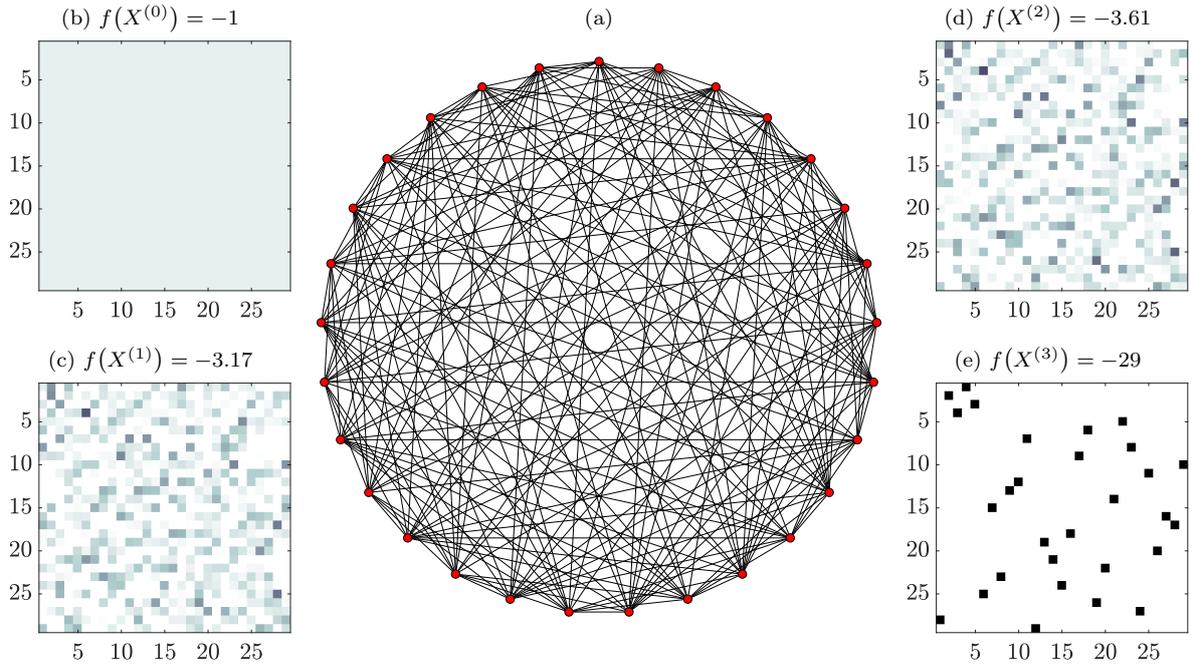

We first consider the Paley graph \cite{Paley33} of size $ n = 29 $ shown in Figure~\ref{fig:Paley graph}\ts(a), whose automorphism group contains $ 406 $ permutations, see \cite{Jones20} for a detailed description of the automorphism groups of Paley graphs. The eigenvalues of the graph are $ \lambda_A = \left[14, \frac{1}{2}\sqrt{29} - \frac{1}{2}, -\frac{1}{2}\sqrt{29} - \frac{1}{2}\right] $ with multiplicities $ \mu_A = [1, 14, 14] $. That is, $ \rank(H) = 448 $, whereas $ n^2 = 841 $. Despite the many symmetries, the Frank--Wolfe solver quickly converges to a permutation matrix as shown in Figures~\ref{fig:Paley graph}\ts(b)--(e). The number of iterations depends on the random permutation we choose to generate $ \mc[B]{G} $.

\subsection{Biggs--Smith graph}

\begin{figure}
    \centering
    \begin{minipage}[t]{0.32\linewidth}
        \centering
        ~~~~\subfiguretitle{(b)}
        \vspace*{0.5ex}
        ~~~~\includegraphics[width=0.8\linewidth]{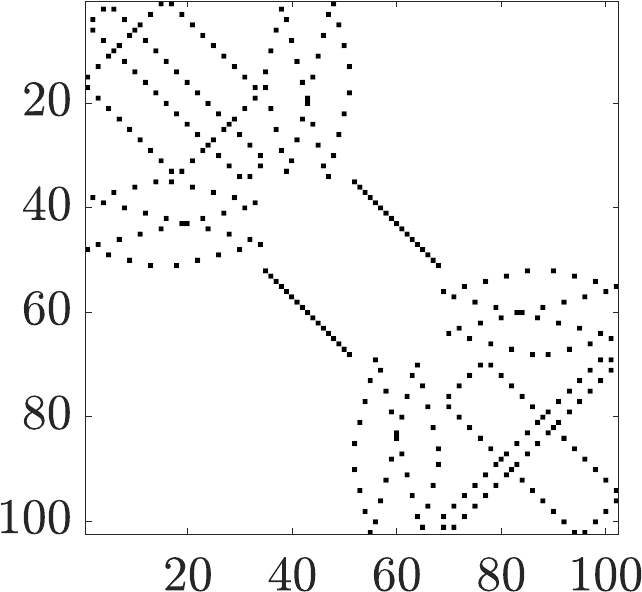} \\[1ex]
        ~~~~\subfiguretitle{(c)}
        \vspace*{0.5ex}
        \includegraphics[width=0.9\linewidth]{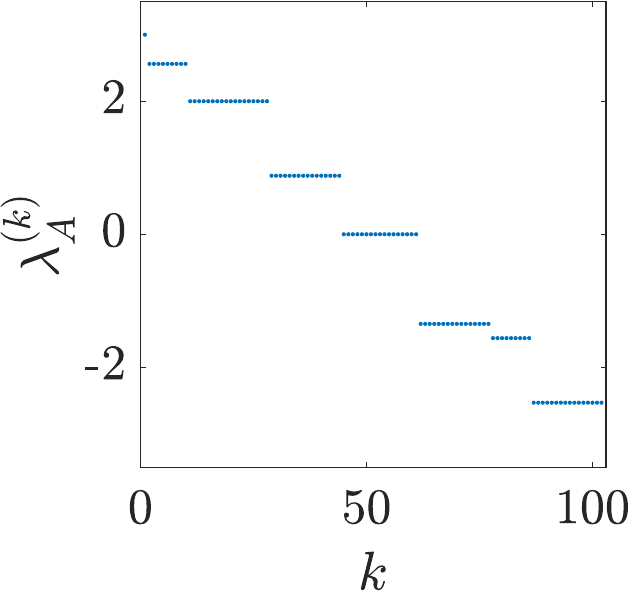}
    \end{minipage}
    \begin{minipage}[t]{0.32\linewidth}
        \centering
        \subfiguretitle{(a)}
        \vspace*{1ex}
        \rotatebox{90}{%
        \begin{tikzpicture}[
            dot/.style={draw,circle,minimum size=1.2mm,inner sep=0pt,outer sep=0pt,fill=red}
            ]
        \def\factor{8.2}
        \coordinate [dot] (1) at (\factor*0.,\factor*0.114553);
        \coordinate [dot] (2) at (\factor*0.,\factor*0.452007);
        \coordinate [dot] (3) at (\factor*0.00427368,\factor*0.160732);
        \coordinate [dot] (4) at (\factor*0.00427368,\factor*0.405828);
        \coordinate [dot] (5) at (\factor*0.0126821,\factor*0.0699353);
        \coordinate [dot] (6) at (\factor*0.0126821,\factor*0.496624);
        \coordinate [dot] (7) at (\factor*0.0249047,\factor*0.202225);
        \coordinate [dot] (8) at (\factor*0.0249047,\factor*0.364335);
        \coordinate [dot] (9) at (\factor*0.0406195,\factor*0.0329463);
        \coordinate [dot] (10) at (\factor*0.0406195,\factor*0.53366);
        \coordinate [dot] (11) at (\factor*0.0591831,\factor*0.23347);
        \coordinate [dot] (12) at (\factor*0.0591831,\factor*0.333089);
        \coordinate [dot] (13) at (\factor*0.0799987,\factor*0.00850105);
        \coordinate [dot] (14) at (\factor*0.0799987,\factor*0.558059);
        \coordinate [dot] (15) at (\factor*0.102422,\factor*0.250242);
        \coordinate [dot] (16) at (\factor*0.102422,\factor*0.316318);
        \coordinate [dot] (17) at (\factor*0.125581,\factor*0.);
        \coordinate [dot] (18) at (\factor*0.125581,\factor*0.56656);
        \coordinate [dot] (19) at (\factor*0.148739,\factor*0.250242);
        \coordinate [dot] (20) at (\factor*0.148739,\factor*0.316318);
        \coordinate [dot] (21) at (\factor*0.171163,\factor*0.00850105);
        \coordinate [dot] (22) at (\factor*0.171163,\factor*0.558059);
        \coordinate [dot] (23) at (\factor*0.191978,\factor*0.23347);
        \coordinate [dot] (24) at (\factor*0.191978,\factor*0.333089);
        \coordinate [dot] (25) at (\factor*0.210541,\factor*0.0329463);
        \coordinate [dot] (26) at (\factor*0.210541,\factor*0.53366);
        \coordinate [dot] (27) at (\factor*0.22621,\factor*0.202225);
        \coordinate [dot] (28) at (\factor*0.22621,\factor*0.364335);
        \coordinate [dot] (29) at (\factor*0.238479,\factor*0.0699353);
        \coordinate [dot] (30) at (\factor*0.238479,\factor*0.496624);
        \coordinate [dot] (31) at (\factor*0.246888,\factor*0.160732);
        \coordinate [dot] (32) at (\factor*0.246888,\factor*0.405828);
        \coordinate [dot] (33) at (\factor*0.251161,\factor*0.114553);
        \coordinate [dot] (34) at (\factor*0.251161,\factor*0.452007);
        \coordinate [dot] (35) at (\factor*0.421174,\factor*0.0308782);
        \coordinate [dot] (36) at (\factor*0.421174,\factor*0.0624459);
        \coordinate [dot] (37) at (\factor*0.421174,\factor*0.0940129);
        \coordinate [dot] (38) at (\factor*0.421174,\factor*0.125534);
        \coordinate [dot] (39) at (\factor*0.421174,\factor*0.157102);
        \coordinate [dot] (40) at (\factor*0.421174,\factor*0.188624);
        \coordinate [dot] (41) at (\factor*0.421174,\factor*0.220191);
        \coordinate [dot] (42) at (\factor*0.421174,\factor*0.251758);
        \coordinate [dot] (43) at (\factor*0.421174,\factor*0.28328);
        \coordinate [dot] (44) at (\factor*0.421174,\factor*0.314847);
        \coordinate [dot] (45) at (\factor*0.421174,\factor*0.346369);
        \coordinate [dot] (46) at (\factor*0.421174,\factor*0.377936);
        \coordinate [dot] (47) at (\factor*0.421174,\factor*0.409504);
        \coordinate [dot] (48) at (\factor*0.421174,\factor*0.441025);
        \coordinate [dot] (49) at (\factor*0.421174,\factor*0.472592);
        \coordinate [dot] (50) at (\factor*0.421174,\factor*0.504114);
        \coordinate [dot] (51) at (\factor*0.421174,\factor*0.535682);
        \coordinate [dot] (52) at (\factor*0.578829,\factor*0.0308782);
        \coordinate [dot] (53) at (\factor*0.578829,\factor*0.0624459);
        \coordinate [dot] (54) at (\factor*0.578829,\factor*0.0940129);
        \coordinate [dot] (55) at (\factor*0.578829,\factor*0.125534);
        \coordinate [dot] (56) at (\factor*0.578829,\factor*0.157102);
        \coordinate [dot] (57) at (\factor*0.578829,\factor*0.188624);
        \coordinate [dot] (58) at (\factor*0.578829,\factor*0.220191);
        \coordinate [dot] (59) at (\factor*0.578829,\factor*0.251758);
        \coordinate [dot] (60) at (\factor*0.578829,\factor*0.28328);
        \coordinate [dot] (61) at (\factor*0.578829,\factor*0.314847);
        \coordinate [dot] (62) at (\factor*0.578829,\factor*0.346369);
        \coordinate [dot] (63) at (\factor*0.578829,\factor*0.377936);
        \coordinate [dot] (64) at (\factor*0.578829,\factor*0.409504);
        \coordinate [dot] (65) at (\factor*0.578829,\factor*0.441025);
        \coordinate [dot] (66) at (\factor*0.578829,\factor*0.472592);
        \coordinate [dot] (67) at (\factor*0.578829,\factor*0.504114);
        \coordinate [dot] (68) at (\factor*0.578829,\factor*0.535682);
        \coordinate [dot] (69) at (\factor*0.748799,\factor*0.114553);
        \coordinate [dot] (70) at (\factor*0.748799,\factor*0.452007);
        \coordinate [dot] (71) at (\factor*0.75307,\factor*0.160732);
        \coordinate [dot] (72) at (\factor*0.75307,\factor*0.405828);
        \coordinate [dot] (73) at (\factor*0.761481,\factor*0.0699353);
        \coordinate [dot] (74) at (\factor*0.761481,\factor*0.496624);
        \coordinate [dot] (75) at (\factor*0.773744,\factor*0.202225);
        \coordinate [dot] (76) at (\factor*0.773744,\factor*0.364335);
        \coordinate [dot] (77) at (\factor*0.789418,\factor*0.0329463);
        \coordinate [dot] (78) at (\factor*0.789418,\factor*0.53366);
        \coordinate [dot] (79) at (\factor*0.808026,\factor*0.23347);
        \coordinate [dot] (80) at (\factor*0.808026,\factor*0.333089);
        \coordinate [dot] (81) at (\factor*0.828839,\factor*0.00850105);
        \coordinate [dot] (82) at (\factor*0.828839,\factor*0.558059);
        \coordinate [dot] (83) at (\factor*0.851219,\factor*0.250242);
        \coordinate [dot] (84) at (\factor*0.851219,\factor*0.316318);
        \coordinate [dot] (85) at (\factor*0.874421,\factor*0.);
        \coordinate [dot] (86) at (\factor*0.874421,\factor*0.56656);
        \coordinate [dot] (87) at (\factor*0.89758,\factor*0.250242);
        \coordinate [dot] (88) at (\factor*0.89758,\factor*0.316318);
        \coordinate [dot] (89) at (\factor*0.919959,\factor*0.00850105);
        \coordinate [dot] (90) at (\factor*0.919959,\factor*0.558059);
        \coordinate [dot] (91) at (\factor*0.940817,\factor*0.23347);
        \coordinate [dot] (92) at (\factor*0.940817,\factor*0.333089);
        \coordinate [dot] (93) at (\factor*0.959381,\factor*0.0329463);
        \coordinate [dot] (94) at (\factor*0.959381,\factor*0.53366);
        \coordinate [dot] (95) at (\factor*0.975055,\factor*0.202225);
        \coordinate [dot] (96) at (\factor*0.975055,\factor*0.364335);
        \coordinate [dot] (97) at (\factor*0.987318,\factor*0.0699353);
        \coordinate [dot] (98) at (\factor*0.987318,\factor*0.496624);
        \coordinate [dot] (99) at (\factor*0.995729,\factor*0.160732);
        \coordinate [dot] (100) at (\factor*0.995729,\factor*0.405828);
        \coordinate [dot] (101) at (\factor*1.,\factor*0.114553);
        \coordinate [dot] (102) at (\factor*1.,\factor*0.452007);
        \path[-,line width=0.05mm] (1) edge (15)
                 (1) edge (17)
                 (1) edge (48)
                 (2) edge (4)
                 (2) edge (6)
                 (2) edge (38)
                 (3) edge (13)
                 (3) edge (19)
                 (3) edge (47)
                 (4) edge (8)
                 (4) edge (39)
                 (5) edge (11)
                 (5) edge (21)
                 (5) edge (49)
                 (6) edge (10)
                 (6) edge (37)
                 (7) edge (9)
                 (7) edge (23)
                 (7) edge (46)
                 (8) edge (12)
                 (8) edge (40)
                 (9) edge (25)
                 (9) edge (50)
                 (10) edge (14)
                 (10) edge (36)
                 (11) edge (27)
                 (11) edge (45)
                 (12) edge (16)
                 (12) edge (41)
                 (13) edge (29)
                 (13) edge (51)
                 (14) edge (18)
                 (14) edge (35)
                 (15) edge (31)
                 (15) edge (44)
                 (16) edge (20)
                 (16) edge (42)
                 (17) edge (33)
                 (17) edge (35)
                 (18) edge (22)
                 (18) edge (51)
                 (19) edge (33)
                 (19) edge (43)
                 (20) edge (24)
                 (20) edge (43)
                 (21) edge (31)
                 (21) edge (36)
                 (22) edge (26)
                 (22) edge (50)
                 (23) edge (29)
                 (23) edge (42)
                 (24) edge (28)
                 (24) edge (44)
                 (25) edge (27)
                 (25) edge (37)
                 (26) edge (30)
                 (26) edge (49)
                 (27) edge (41)
                 (28) edge (32)
                 (28) edge (45)
                 (29) edge (38)
                 (30) edge (34)
                 (30) edge (48)
                 (31) edge (40)
                 (32) edge (34)
                 (32) edge (46)
                 (33) edge (39)
                 (34) edge (47)
                 (35) edge (52)
                 (36) edge (53)
                 (37) edge (54)
                 (38) edge (55)
                 (39) edge (56)
                 (40) edge (57)
                 (41) edge (58)
                 (42) edge (59)
                 (43) edge (60)
                 (44) edge (61)
                 (45) edge (62)
                 (46) edge (63)
                 (47) edge (64)
                 (48) edge (65)
                 (49) edge (66)
                 (50) edge (67)
                 (51) edge (68)
                 (52) edge (85)
                 (52) edge (90)
                 (53) edge (81)
                 (53) edge (94)
                 (54) edge (77)
                 (54) edge (98)
                 (55) edge (73)
                 (55) edge (102)
                 (56) edge (69)
                 (56) edge (100)
                 (57) edge (71)
                 (57) edge (96)
                 (58) edge (75)
                 (58) edge (92)
                 (59) edge (79)
                 (59) edge (88)
                 (60) edge (83)
                 (60) edge (84)
                 (61) edge (80)
                 (61) edge (87)
                 (62) edge (76)
                 (62) edge (91)
                 (63) edge (72)
                 (63) edge (95)
                 (64) edge (70)
                 (64) edge (99)
                 (65) edge (74)
                 (65) edge (101)
                 (66) edge (78)
                 (66) edge (97)
                 (67) edge (82)
                 (67) edge (93)
                 (68) edge (86)
                 (68) edge (89)
                 (69) edge (99)
                 (69) edge (101)
                 (70) edge (76)
                 (70) edge (78)
                 (71) edge (97)
                 (71) edge (101)
                 (72) edge (74)
                 (72) edge (80)
                 (73) edge (95)
                 (73) edge (99)
                 (74) edge (82)
                 (75) edge (93)
                 (75) edge (97)
                 (76) edge (84)
                 (77) edge (91)
                 (77) edge (95)
                 (78) edge (86)
                 (79) edge (89)
                 (79) edge (93)
                 (80) edge (88)
                 (81) edge (87)
                 (81) edge (91)
                 (82) edge (90)
                 (83) edge (85)
                 (83) edge (89)
                 (84) edge (92)
                 (85) edge (87)
                 (86) edge (94)
                 (88) edge (96)
                 (90) edge (98)
                 (92) edge (100)
                 (94) edge (102)
                 (96) edge (102)
                 (98) edge (100);
        \end{tikzpicture}}
    \end{minipage}
    \begin{minipage}[t]{0.32\linewidth}
        \centering
        \subfiguretitle{(d) $ f\big(X^{(0)}\big) = -1 $}
        \vspace*{0.5ex}
        \includegraphics[width=0.8\linewidth]{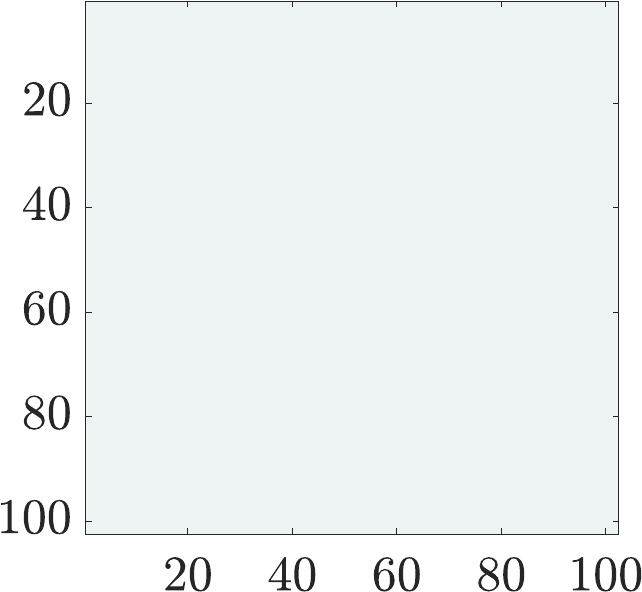} \\[1ex]
        \subfiguretitle{(e) $ f\big(X^{(1)}\big) = -102 $}
        \vspace*{0.5ex}
        \includegraphics[width=0.8\linewidth]{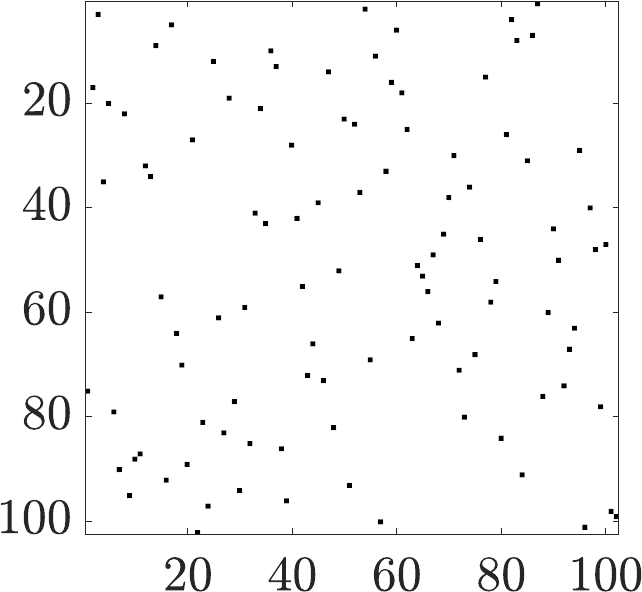}
    \end{minipage}
    \caption{(a) Visualization of the Biggs--Smith graph $ \mc[A]{G} $ using the layout from \cite{Weisstein23}. The graph $ \mc[B]{G} $ is again defined to be a random permutation of $ \mc[A]{G} $. (b) Adjacency matrix $ A $. (c) Spectrum of $ A $. (d)--(e) Convergence of the Frank--Wolfe solver to a permutation matrix after just one step. The initial matrix $ X^{(0)} $ is again constant.}
    \label{fig:Biggs-Smith graph}
\end{figure}
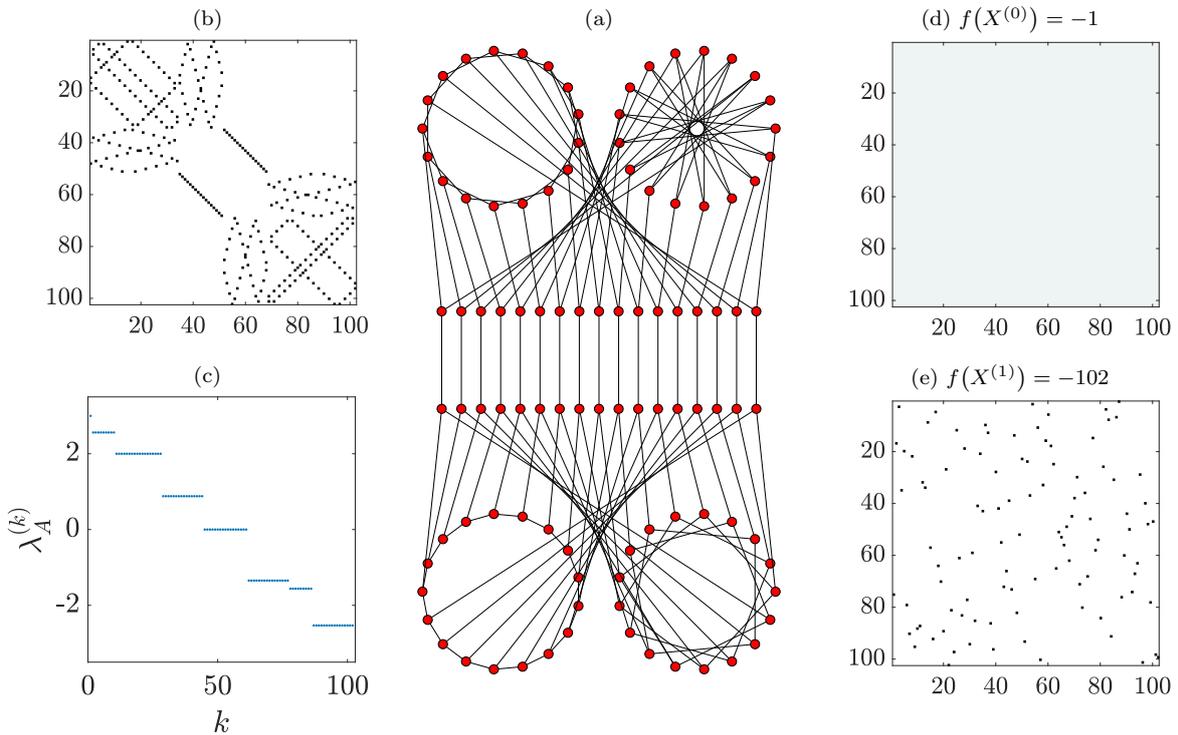

The automorphism group of the Biggs--Smith graph \cite{BS71} of size $ n = 102 $, see Figure~\ref{fig:Biggs-Smith graph}\ts(a), contains 2448 elements. The adjacency matrix $ A $ is shown in Figure~\ref{fig:Biggs-Smith graph}\ts(b) and its spectrum in Figure~\ref{fig:Biggs-Smith graph}\ts(c). The multiplicities of the eigenvalues are $ \mu_A = [1, 9, 18, 16, 17, 16, 9, 16] $, which implies that $ \rank(H) = 8860 $. The Frank--Wolfe solver converges after just one iteration to a permutation matrix as shown in Figures~\ref{fig:Biggs-Smith graph}\ts(d)--(e).

\section{Conclusion}

The main advantage of the proposed approach is that the graph isomorphism problem is turned into a continuous optimization problem, which in general tend to be easier to solve than combinatorial optimization problems. Although orthogonal and doubly stochastic relaxations have been considered before, we extended these techniques to strongly regular graphs and showed how symmetries and repeated eigenvalues increase the dimensions of the spaces of feasible solutions. We believe that continuous relaxations will not only open new avenues for deriving efficient numerical algorithms for the graph isomorphism problem, but also for gaining insights into the geometric properties of the sets of feasible solutions and the inherent symmetries of graphs.

Although the Frank--Wolfe solver typically converges within a few iterations, the algorithm sometimes gets stuck in local minima, which are often convex combinations of just a few optimal permutation matrices. We also implemented SQP-based solvers, but they suffer from the same problem. Small random perturbations of the initial condition in a feasible direction seem to alleviate this problem. Also using different (i.e., more aggressive) penalty functions, adding barrier functions, or continuation methods might help avoid local minima. A detailed analysis of the convergence properties and improved search-direction strategies are beyond the scope of this paper and will be considered in future work.

Further open problems include: How sensitive is the algorithm with respect to numerical errors and is it possible to develop sparse solvers? What is the worst-time complexity of the algorithm and for which types of graphs is the problem guaranteed to be solvable in polynomial time? Can we exploit properties of the orthogonal and convex relaxations at the same time? Another interesting question is whether this approach can be extended to the graph matching problem by including not only the vectors contained in the null space of $ H $. The derived optimization problem formulation might also be amenable to quantum annealers, provided that the constrained problem can be rewritten as a quadratic unconstrained binary optimization (QUBO) problem.

\section*{Acknowledgments}

P.G.\ was supported by the \emph{QuantERA II} programme funded by the European Union’s Horizon~2020 research and innovation programme under Grant Agreement No.\ 101017733. We would like to thank the reviewers for the helpful comments and suggestions.

\section*{Data availability}

The Matlab code and examples that support the findings presented in this paper are available at \url{https://github.com/sklus/FrankWolfeGI}.

\bibliographystyle{unsrturl}
\bibliography{GI}

\end{document}